\documentclass[acmsmall,screen]{acmart}

\setcopyright{cc}
\setcctype{by}
\acmDOI{10.1145/3763067}
\acmYear{2025}
\acmJournal{PACMPL}
\acmVolume{9}
\acmNumber{OOPSLA2}
\acmArticle{289}
\acmMonth{10}
\received{2025-03-25}
\received[accepted]{2025-08-12}

\usepackage{base}
\usepackage{macros}

\begin{document}
\title{\Rtrace: A Semantics for Understanding React Hooks}
\subtitle{An Operational Semantics and a Visualizer for Clarifying React Hooks}

\author{Jay Lee}
\orcid{0000-0002-2224-4861}
\email{jhlee@ropas.snu.ac.kr}

\author{Joongwon Ahn}
\orcid{0009-0007-0055-7297}
\email{jwahn@ropas.snu.ac.kr}

\author{Kwangkeun Yi}
\orcid{0009-0007-5027-2177}
\email{kwang@ropas.snu.ac.kr}

\affiliation{%
  \department{Department of Computer Science and Engineering}
  \institution{Seoul National University}
  \city{Seoul}
  \country{Korea}}



\citestyle{acmauthoryear}

\begin{abstract}
React has become the most widely used web front-end framework, enabling the creation of user interfaces in a declarative and compositional manner.
Hooks are a set of APIs that manage side effects in function components in React.
However, their semantics are often seen as opaque to developers, leading to UI bugs.
We introduce \Rtrace, a formalization of the semantics of the essence of React Hooks, providing a semantics that clarifies their behavior.
We demonstrate that our model captures the behavior of React, by theoretically showing that it embodies essential properties of Hooks and empirically comparing our \Rtrace-definitional interpreter against a test suite.
Furthermore, we showcase a practical visualization tool based on the formalization to demonstrate how developers can better understand the semantics of Hooks.
\end{abstract}

\maketitle

\section{Introduction}\label{sec:intro}
\subsection{Overview of React}
React (\href{https://react.dev/}{react.dev}) is the most widely used framework for developing web-based graphical user interfaces (GUI) today \citep{Sta24Stack}.
React developers structure the UI as a tree of \emph{components,} where each component serves as a \emph{functional} specification that declaratively defines a portion of the UI and its behavior.
To declare user interactions in a component, developers use \emph{Hooks}, a collection of library functions that ``hooks into'' the runtime to manage state and handle side effects within components.

A React component is a function that serves as a declarative specification of a UI, which is read by the React runtime to render the interface.
The process of ``reading'' this specification is nothing special---the runtime directly invokes the component, which is a plain function.
When these components are called with appropriate arguments, they return a data structure representing how the component should appear on the screen.
Ideally, when a component is called with the same arguments, it should render the same interface \citep{Cac16React}.

This functional approach differs significantly from traditional object-oriented UI frameworks.
In object-oriented frameworks like Flutter (\href{https://flutter.dev/}{flutter.dev}), Qt (\href{https://www.qt.io/}{qt.io}), UIKit (\href{https://developer.apple.com/documentation/uikit/}{developer.apple.com}), or even legacy versions of React \citep{Met25Component,MadLhoTip20Semantics}, components are modeled with classes, where each view directly corresponds to an instance of a class.
Developers define different methods to be called at different lifecycle phases---such as initialization, state update, etc.---creating a clear separation.

React, however, introduces a different model: a single function component plays multiple roles---the same component is called for both creating a new view and updating an existing view.
To enable this, React switches between different implementations of the same Hook at runtime depending on the rendering phase.

\begin{wrapfigure}[4]{l}{0.387\linewidth}
  \vspace*{-1.2\intextsep}
  \begin{reactcode}
function Counter({ x }) {
  const [s, setS] = useState(() => x);
  const h = () => setS((s) => s+1);
  return <button onClick={h}> {s}
    </button>; }
  \end{reactcode}
\end{wrapfigure}
\noindent The \reacttt{Counter} component on the left demonstrates the \reacttt{useState} Hook for managing state.
This simple component increments its state~\reacttt{s} each time the button is clicked.
When first rendered, \reacttt{useState} returns an initial value from the initializer function~\reacttt{() => x}.
Here,~\reacttt{s} captures the current state, and \reacttt{setS}~accepts an updater to set a new state.
In subsequent renders, \reacttt{useState} automagically ``remembers'' the previous state and returns it instead.
\reacttt{button} is correctly rendered with new~\reacttt{s}, by efficiently diffing the returned \reacttt{button}---this is called \emph{reconciliation} \citep{Met25Reconciliation}.
This example illustrates how even a simple component like \reacttt{Counter} has different semantics across lifecycle phases.

Another peculiarity is that developers \emph{cannot} use JavaScript (JS) variables and direct assignments to manage a component's state.
In the above example, assigning a new value to~\reacttt{s} does not update the state of~\reacttt{Counter}---\reacttt{setS} must be used to update the state.

\begin{wrapfigure}[8]{r}{0.44\linewidth}
  \vspace*{-.8\intextsep}
  \begin{reactcode}
function Cond() {
  const [b, setB] = useState(() => false);
  if (b) {
    const [s, setS] = useState(() => 0);
    return s;
  } else {
    const h = () => setB(b => !b);
    return <button onClick={h}> Show
      </button>; } }
  \end{reactcode}
\end{wrapfigure}
Moreover, developers must adhere to specific rules imposed by React, collectively known as the \emph{Rules of React} \citep{Met25Rules}.
These rules represent invariants that components must maintain for React to properly manage the UI.
The \reacttt{Cond} component on the right violates one of these rules on line~4 by calling \reacttt{useState} conditionally inside the \reacttt{if}~branch.
React requires that all Hooks be called unconditionally at the top-level of a component, ensuring they execute in the same order on every render.
Initially, \reacttt{Cond} will show a button to toggle the Boolean state~\reacttt{b} from \reacttt{false} to \reacttt{true}.
When you click ``Show'' (setting~\reacttt{b} to \reacttt{true}), instead of displaying the initial state of~\reacttt{s} which is~|0|, the UI breaks down at runtime.
This stems from React's implementation detail where each Hook call is identified by its position in a linked list of internal bookkeeping objects.

\subsection{The Problem}
Hooks have subtle semantics, making it challenging to avoid unexpected behaviors.
The rendering process is opaque---it is unclear \emph{how} React registers updates or \emph{when} re-renders occur.
Therefore, React developers must follow specific rules or risk breaking their UI (detailed in~\cref{sec:pit}).
Worse yet, the host language is oblivious to React's rules, so it's up to programmers to adhere to them.
Without a clear understanding of the underlying mechanisms, developers must rely on informal resources that typically provide high-level guides rather than explanations of the render semantics, leaving even experienced developers vulnerable to misconceptions about how Hooks actually work.

\subsection{Our Solution}
We present \Rtrace, an operational semantics of React Hooks:
\begin{itemize}
  \item \textbf{\Rtrace{} is a semantics for clarifying the behavior of React Hooks.}
    \Rtrace{} is \emph{high-level enough} to abstract away implementation details and \emph{precise enough} to aid developers in reasoning about rendering behavior~(\cref{subsec:example}).
    Our semantics helps developers clearly understand confusing UI bugs like infinite re-rendering~(\cref{sec:pit}).

  \item \textbf{\Rtrace{} is a foundation for building semantic-based tools.}
    \Rtrace{} provides a semantics for React Hooks~(\cref{sec:sem}), which is the first step in designing any semantic-based tool, such as an abstract interpreter \citep{CouCou77Abstract,RivYi20Introduction}.
    As an example for a semantic-based tool based on \Rtrace, we present an interactive tool that explains and visualizes the behavior of React programs~(\cref{sec:impl}).

  \item \textbf{\Rtrace{} captures the essence of React.}
    We prove that \Rtrace{} conforms to the key properties of React according to the documentation, as well as presenting an empirical conformance test suite~(\cref{sec:rules}).
\end{itemize}

\paragraph{Organization}
We introduce \reacttt{useState} and \reacttt{useEffect} and describe challenges in understanding them in~\cref{sec:over}, and illustrate common pitfalls from their interactions in~\cref{sec:pit}.
We formalize React Hook semantics in~\cref{sec:sem}, present a \Rtrace-definitional interpreter and visualizer in~\cref{sec:impl}, and demonstrate conformance with React in~\cref{sec:rules}.
We discuss extensions and implications in~\cref{sec:discussion}, compare with related work in~\cref{sec:related}, and conclude with future work in~\cref{sec:conclusion}.

\section{Peculiarities of React Hooks}\label{sec:over}
In this section, we examine the peculiarities of the \reacttt{useState}~(\cref{subsec:useState}) and \reacttt{useEffect}~(\cref{subsec:useEffect}) Hooks.

We introduce a simple applicative language for \Rtrace{} that captures the essence of React Hooks and use it throughout the paper to illustrate the examples and the semantics.
Thus we decouple ourselves from the subtleties of the complex semantics of JS \citep{RyuPar24JavaScript}, which is beyond the scope of our research.

\paragraph{Syntax}
The syntax of \Rtrace{} is given in \cref{fig:syntax}.
A React program~$P$ is a sequence of \emph{component definitions}~$D$, followed by the \emph{main expression}~$e$.
A component definition~$D = \cfunc{C}{x}{e}$ defines a component named~$C$ that accepts parameter~$x$ bound in body~$e$.
A component name~$C$ is capitalized to distinguish it from other identifiers.
Expressions~$e$ are mostly standard, except for the array view~$[\Overline{e}]$ for nested views and the Hooks.
The \reacttt{useState} and \reacttt{useEffect} Hooks can be used only at the top level of the component bodies; they cannot be used as a sub-expression of other expressions---this is enforced by the Rules of React \citep{Met25Hooks}.
We check this syntactically during parsing in the \Rtrace{} interpreter~(\cref{sec:impl}).
Each \reacttt{useState} Hook is labeled uniquely with a natural number~$\ell$.
A view represents the structure of the UI that evaluates into either a constant, a closure, or an array view.
An array view can represent a nested structure, modeling the JSX syntax \citep{Met22JSX}.
A closure used as a view represents an event handler, e.g., a button or a text input.
Metavariable~$\oplus$ stands for total binary operations over integers.
$\print{e}$ prints to the console.
We only consider programs~$P$ whose main expressions evaluate to a view.

\begin{figure}[t]
  \centering
  \begin{minipage}{\textwidth}
    \begin{flalign*}
      x, C &\in \dom{Var} & n &\in \mathbb{Z} & \ell &\in \mathbb{N} &
      \dom{Prog}\ P &\Coloneq \Overline{D}\ e & \dom{ComDef}\ D &\Coloneq \cfunc{C}{x}{e}
    \end{flalign*}
  \end{minipage}
  \begin{bnf}[rccll]
    e : \dom{Exp} ::=
    | () // true // false // $n$ // $x$ // $C$ // $e \oplus e$ // $[\Overline{e}]$ // $\print{e}$ // $\cond{e}{e}{e}$ // $\seq{e}{e}$
    | $\func{x}{e}$ // $\app{e}{e}$ // $\letbind{x}{e}{e}$ // $\stbind{x}{e}{e}$ // $\eff{e}$
  \end{bnf}
  \caption{Syntax of \Rtrace.}\label{fig:syntax}
\end{figure}

We model only the \reacttt{useState} and \reacttt{useEffect} Hooks in our language, as they are the most essential Hooks that are closely coupled with the rendering of components.
\begin{itemize}
  \item The \reacttt{useState} Hook provides a way to manage state in components, making it a cornerstone for defining user interactions in components.
    \reacttt{useState} returns a pair of values: the current state~$x$ and a setter function~$x_{\textsf{set}}$ that updates the state.
  \item The \reacttt{useEffect} Hook enables running arbitrary code~$e$ after a component is rendered.
    This Hook makes it useful to synchronize with the ``outside world,'' but can lead to complex behavior when it triggers a re-render.
    The expression passed to \reacttt{useEffect},~$e$, is called an \emph{Effect} \citep{Met25Syn}.
\end{itemize}
Actually, as we will see in~\cref{sec:pit}, the common pitfalls of Hooks are caused by misusing \reacttt{useState} and \reacttt{useEffect} in combination.
We now explain how these Hooks are used in React components in~\cref{subsec:useState,subsec:useEffect}, before presenting their exact semantics in~\cref{sec:sem}.

\subsection{The Peculiarity of \textbf{\texttt{useState}}}\label{subsec:useState}
Among the top~45 most frequent StackOverflow questions with the tag \textsf{[react-hooks]}, three---including the most frequent question by \citet{Pra19useState}---are about the semantics and timing of the \reacttt{useState} Hook and its state updates \citep{Pra19useState,vad18React,Tan19React}.

\begin{wrapfigure}[4]{l}{0.23\linewidth}
  \vspace*{-\intextsep}
  \begin{reactcode}
let Counter x =
  let s := 0 in
  [s, button (fun _ ->
    s := s+2)];;
Counter 0
  \end{reactcode}
\end{wrapfigure}
\noindent Consider a counter that increments its value by two when clicked.
If our language had mutable variables, one might na\"ively write the code on the left by defining a mutable variable~\reacttt{s} and adding two to it when the button is clicked.
However, this approach presents two significant problems:
\begin{enumerate}
  \item The variable~\reacttt{s} would be rebound to~|0| every time the component is invoked, preventing the runtime from re-reading the component without overwriting the previous state.
  \item Even if the previous problem were somehow solved, to determine when to update the UI to synchronize with the states, the runtime would need to track all mutable variables in each component, which would be inefficient in practice.
\end{enumerate}

The \reacttt{useState} Hook addresses these issues by providing functionality analogous to mutable variables in imperative languages, yet with a functional approach.
Instead of directly modifying state through assignment, the state is managed by the React runtime and can only be updated by the setter function returned by the Hook.
The \reacttt{useState} Hook allows the runtime to be aware of the state and handle its bookkeeping, therefore addressing the aforementioned two problems.

\begin{wrapfigure}[5]{r}{0.33\linewidth}
  \vspace*{-\intextsep}
  \begin{reactcode}
let Counter x =
  let (s, setS) = useState x in
  [s, button (fun _ ->
    setS (fun s -> s+1);
    setS (fun s -> s+1))];;
Counter 0
  \end{reactcode}
\end{wrapfigure}
Using \reacttt{useState}, we can correctly implement a functioning \reacttt{Counter} as shown on the right.
\reacttt{useState} returns a pair of values: \reacttt{s}~that stores the current state and the setter function~\reacttt{setS} that accepts an updater function.
Initially, \reacttt{s}~holds the value of~\reacttt{x}.
On subsequent renders, React is able to provide the correct current state without reinitializing to~\reacttt{x}.

When a button is pressed, the callback passed to it calls \reacttt{setS} twice with the updater~|fun s->s+1|.
Unlike direct assignment to~\reacttt{s} in the previous example, this does not immediately update~\reacttt{s}, but queues the update \citep{Met25Queueing}.
The update is processed by the runtime during the next render, ensuring that the view is re-rendered with the updated state.
This raises a question:
\begin{quote}
  \emph{Precisely \emph{when} in the runtime is this update handled?}
\end{quote}

To investigate when the callback provided to \reacttt{setS} is executed, we can add diagnostic \reacttt{print}s:
\begin{center}
  \begin{minipage}{.5\linewidth}
    \begin{reactcode}
let Counter x =
  print "Counter";
  let (s, setS) = useState x in
  print "Return";
  [s, button (fun _ ->
    setS (fun s -> s+1);
    setS (fun s -> print "Update"; s+1))];;
Counter 0
    \end{reactcode}
  \end{minipage}
  \begin{minipage}{.3\linewidth}
    \begin{console}[frame=topline]
Counter
Return
    \end{console}
    {\footnotesize\hspace{2em}\reacttt{button} \faHandPointUp[regular]}
    \begin{console}[frame=bottomline, firstnumber=3]
Counter
Update
Return
    \end{console}
  \end{minipage}
\end{center}
The console output shows that \verb+Counter+ and \verb+Return+ are printed during the initial rendering.
When the button is clicked, \verb+Update+ is printed \emph{after} \verb+Counter+ and \emph{before} \verb+Return+---a behavior that is not immediately apparent from the code alone.

In fact, the official React documentation does not specify exactly when the queued updates are processed, other than that they are processed sometime during the next render.
We will clearly define how this happened by providing a formal semantics of \Rtrace{} in~\cref{sec:sem}.

\subsection{The Peculiarity of \textbf{\texttt{useEffect}}}\label{subsec:useEffect}
While \reacttt{useState} hooks into the React runtime to manage component state, \reacttt{useEffect} hooks into the rendering lifecycle of components.
Every time a component is rendered, an Effect---a suspended computation or a thunk---provided to the \reacttt{useEffect} Hook is executed.
Effect allows developers to run arbitrary logic after the render.
Note that in React, it is possible to specify a set of variables called \emph{dependencies,} so that the Effect runs only if those variables have changed.
We only consider the simplest form of \reacttt{useEffect} in our language, where the Effect is run unconditionally after each render, and it is a straightforward extension to support dependencies.

\begin{wrapfigure}[7]{l}{0.35\linewidth}
  \vspace*{-\intextsep}
  \begin{reactcode}
let Counter x =
  let (s, setS) = useState x in
  useEffect (print (if s mod 2 = 0
    then "Even" else "Odd"));
  [s, button (fun _ ->
    setS (fun s -> s+1);
    setS (fun s -> s+1))];;
Counter 0
  \end{reactcode}
\end{wrapfigure}
\noindent Building upon the previous \reacttt{Counter} example, we can log \verb+Even+ or \verb+Odd+ depending on the value of the counter.
The \reacttt{useEffect} Hook makes this straightforward, as shown on the left.
This implementation prints either \verb+Even+ or \verb+Odd+ to the console after each render based on the parity of the counter.
This ``delayed'' logging is possible as an Effect expression passed to \reacttt{useEffect} is unevaluated\footnote{In React/JS, an Effect needs to be wrapped inside a callback.} until after the component has rendered.

While a simple \reacttt{print} as an Effect may seem trivial, this amounts to synchronizing with the outside world, which is the very purpose of the \reacttt{useEffect} Hook.
In real life, this can be a call to some logging system or user analytics system.

To understand a more peculiar aspect of \reacttt{useEffect}, consider the following \reacttt{SelfCounter} where an Effect creates an autonomous rendering cycle:
\begin{center}
  \begin{minipage}{.5\linewidth}
    \begin{reactcode}
let SelfCounter x =
  let (s, setS) = useState x in
  print s;
  useEffect (
    print "Effect";
    if s < 3 then
      setS (fun s -> s + 1));
  print "Return";
  [s];;
SelfCounter 0
    \end{reactcode}
  \end{minipage}
  \begin{minipage}{.3\linewidth}
    \begin{console}
0
Return
Effect
1
Return
Effect
2
Return
Effect
3
Return
Effect
    \end{console}
  \end{minipage}
\end{center}
This example demonstrates that Effects can create render cycles \emph{without any user interaction.}
Initially, the value of~\verb+s+,~0 is printed, followed by \verb+Return+ and \verb+Effect+ messages, showing that the Effect runs after the initial render.
After the Effect runs, it updates the state when~\verb+s < 3+, triggering a new render automatically.
The component autonomously increments from~0 to~3 without any user interaction.
When~\verb+s+ reaches~3, the Effect still runs (as shown by the final \verb+Effect+ output), but no further updates are queued.

This pattern creates a dangerous pitfall:
\begin{quote}
  \emph{Developers can unwittingly trigger excess render cycles with seemingly innocent Effects.}
\end{quote}
Our formal semantics in~\cref{sec:sem} captures these patterns, helping developers reason about when and how these cycles occur in their programs.
We explore how this pitfall leads to bugs in~\cref{subsubsec:infloop,subsec:unnec}.

\section{Re-Rendering Pitfalls with React Hooks}\label{sec:pit}
Having examined the peculiarities of \reacttt{useState} and \reacttt{useEffect} in~\cref{sec:over}, we now demonstrate how their interaction can lead to common render-related bugs.
These pitfalls are challenging as they stem from the opaque render semantics.
These issues motivate our formal semantics~(\cref{sec:sem}).

\subsection{Infinite Re-Rendering}\label{subsec:inf}
Infinite re-rendering is perhaps the most catastrophic bug one can encounter when using Hooks.
There are two different problems in this category of bugs: an infinite render loop due to always setting a different state in a \reacttt{useEffect}~(\cref{subsubsec:infloop}) and an infinite re-evaluation of a component body due to a top-level call to a setter function~(\cref{subsubsec:topset}).

\subsubsection{Infinite Render Loop with an Effect}\label{subsubsec:infloop}
Using \reacttt{useEffect} can easily lead to an infinite render loop---searching StackOverflow with the query `\textsf{"useEffect" "infinite"}' returns more than 1600~results \citep{Jay25Archive}.
We describe the essence of the issue here.

Infinite render loop occurs because setting state within an Effect triggers the runtime to \dec{Check}\footnote{We use \dec{red boldfaced sans-serif} to emphasize a decision, whose semantic meaning is formally introduced in~\cref{subsec:dom}.} if the state has changed, which is implemented by invoking the component body again.
If the state whose setter function is called is actually modified, the component re-renders and decides to run the \dec{Effect} again.
Essentially, this \dec{Check}-\dec{Effect} decision cycle creates a render loop.
In the following, we show a basic example that renders infinitely many times due to this mechanism.

\begin{wrapfigure}[4]{l}{0.35\linewidth}
  \vspace{-1.2\intextsep}
  \begin{reactcode}
let Inf x =
  let (s, setS) = useState 0 in
  useEffect (setS (fun s -> s+1));
  s;;
Inf 0
  \end{reactcode}
\end{wrapfigure}
\noindent The component \reacttt{Inf} on the left is a simple example that increments the state by one after each render in an Effect.
Upon each call to \reacttt{setS} while executing an Effect, the component ``remembers'' that its body needs to be scanned again to see if the state has actually changed.
We say that the component decides to \dec{Check} in the next render pass.
Since~\reacttt{s} always increments, the scan results in the runtime re-rendering \reacttt{Inf} indefinitely.

This render loop is similar to the example \verb+SelfCounter+ shown in~\cref{subsec:useEffect}, although this time, the condition that breaks the loop is absent.
A generalized issue is presented in~\cref{subsec:unnec}.

\subsubsection{Top-Level Call to a Setter Function}\label{subsubsec:topset}
Infinite loop caused by a top-level call to a setter function is also a common source of confusion---the most frequent StackOverflow question with both the \textsf{[reactjs]} and \textsf{[infinite-loop]} tag is about this issue \citep{Teh19Updating}.

When a component sets a state in the top-level, the runtime immediately \dec{Check}s the component again by re-evaluating it.
This causes React to discard the returned view and re-read the component.

An active \dec{Check} decision during reading a component is treated specially as a signal to \emph{retry} the component before the render, unlike a \dec{Check} decision while running an Effect.
Therefore, an unconditional top-level call to a setter function causes an infinite loop and is never correct.

\begin{wrapfigure}[4]{r}{0.33\linewidth}
  \vspace{-\intextsep}
  \begin{reactcode}
let Inf2 x =
  let (s, setS) = useState x in
  setS (fun s -> s);
  s;;
Inf2 0
  \end{reactcode}
\end{wrapfigure}
Unlike the infinite render loop bug previously discussed in~\cref{subsubsec:infloop}, \reacttt{Inf2} on the right does not even reach the screen.
|Inf2| simply causes the UI to show a blank screen.
Even if the state is always set to the same value~|0|, |Inf2| falls into an infinite loop.


\subsection{Unnecessary Re-Rendering}\label{subsec:unnec}
A simple example of an unnecessary re-render triggered by the \reacttt{useEffect} Hook is demonstrated in the following component \reacttt{Flicker}.\,
\reacttt{Flicker} starts with an initial state~|0|, and immediately after%
\begin{wrapfigure}[4]{l}{0.34\linewidth}
  \vspace*{-.2\intextsep}
  \begin{reactcode}
let Flicker x =
  let (s, setS) = useState x in
  useEffect (setS (fun _ -> 42));
  s;;
Flicker 0
  \end{reactcode}
\end{wrapfigure}%
the render, it updates the state to~|42|, triggering the runtime to \dec{Check} the component and causes a re-render.
A swift user might even see the transient state~|0| and notice a ``flicker'' in the UI.

While technically a similar problem to the one discussed in~\cref{subsubsec:infloop}, this is a performance issue, in contrast to the other which was an evident bug.
When the component needs to access an external resource in order to set a state, setting a state after the call to \reacttt{useEffect} may be necessary, but in other scenarios, it is a waste of a render cycle to do so.

\begin{wrapfigure}[6]{r}{0.39\linewidth}
  \vspace*{-\intextsep}
  \begin{reactcode}
let Child setS =
  useEffect (setS (fun _ -> false));
  ();;
let Parent b =
  let (s, setS) = useState b in
  if s then Child setS else ();;
Parent true
  \end{reactcode}
\end{wrapfigure}
Note that this issue can happen inter-component as well in a more subtle manner.
In the example on the right, the parent component \reacttt{Parent} decides to render the child component \reacttt{Child} based on the state~\reacttt{s}.
Initially, \reacttt{Parent} renders \reacttt{Child}, passing its setter function \reacttt{setS} to \reacttt{Child}.
After the render, Effect in \reacttt{Child} gets invoked, queuing an update to the state~\reacttt{s} of \reacttt{Parent} to \reacttt{false}.
In the end, \reacttt{Parent} re-renders, but this time, it will not render \reacttt{Child}, as the state~\reacttt{s} becomes \reacttt{false}.

\section{An Operational Semantics for the Essence of React Hooks}\label{sec:sem}
We now provide a formal operational semantics \Rtrace{} that is both high-level enough to abstract away React's implementation details, yet precise enough to reason about rendering behavior~(\cref{sec:rules}).
Up until now, we have implicitly discussed ``values'' using the syntax of \Rtrace, and we formally introduce semantic objects in~\cref{subsec:dom}.
Then we describe the operational semantics of \Rtrace{} in~\cref{subsec:sem}.

There are two layers that constitute the syntax and semantics of React Hooks:
\begin{description}
  \item[Base layer] describes standard computations within components.
    An excerpt from \cref{fig:syntax}:
    \begin{center}
      \begin{bnf}
        e : \dom{Exp} ::= () // true // false // $n$ // $x$  // $e \oplus e$ // $\print{e}$ // $\cond{e}{e}{e}$
          | $\seq{e}{e}$ // $\func{x}{e}$ // $\app{e}{e}$ // $\letbind{x}{e}{e}$ // $\dotsb$
      \end{bnf}
    \end{center}
  \item[Render layer] controls the rendering of components.
    An excerpt from \cref{fig:syntax}:
    \begin{center}
      \begin{bnf}[rccll]
        e : \dom{Exp} ::= $\dotsb$ // $C$ // $[\Overline{e}]$ // $\stbind{x}{e}{e}$ // $\eff{e}$
      \end{bnf}
    \end{center}
\end{description}

The base logic layer is crucial in order to build a rich UI with complex business logic.
For instance, Booleans and conditionals are required to write a toggleable button;
Functions and function applications are necessary to perform actions on behalf of user interaction.

Our choice of the base layer is not the only possible one.
We have chosen a minimal combination of features that can support the essence of Hooks---we can extend our base with strings, objects, recursive functions, etc.
Here, we will not burden ourselves with a layer not too relevant to the render logic.
In fact, all these features are included in our implementation of \Rtrace~(\cref{sec:impl}).

The render logic layer seems quite minimal, but we have seen its subtleties in~\cref{sec:over,sec:pit}.
We include all the base values to be used to represent views, and in addition, we include an array to model nested views and multiple children.
This is not a general purpose array---its elimination is when being used at the returning position of a component.
An example of an array view has been shown in the running examples in~\cref{sec:over}.
Each \reacttt{useState} Hook is (implicitly) labeled so that it is uniquely identified.
The returned value of \reacttt{useState} is a pair of current value~$x$ and the setter function~$x_{\texttt{set}}$.
\reacttt{useEffect} Hooks need not be labeled.

Note that we are only concerned with the timing aspects of the render semantics, not the visual layout; hence interactive UI elements such as a button is identified with its event handler.
As such, \reacttt{button f} is simply desugared into its event handler~\reacttt{f}.

\subsection{Semantic Objects}\label{subsec:dom}
Semantic objects include base values~(\cref{subsubsec:base}) as well as the machinery that models the React runtime and Hooks~(\cref{subsubsec:machinery}).
The complete set of semantic objects is reproduced in~\S A.1 for reference.

\subsubsection{Base Values}\label{subsubsec:base}
The semantic objects for base values are formalized as follows:
\begin{center}
  \begin{minipage}{0.4791\textwidth}
    \begin{center}
      \begin{bnf}[rccll]
        v : $\dom{Val}$ ::= $k$ // $cl$ // $\dotsb$ :
        ;;
        k : $\dom{Const}$ ::= $\<\>$ // $\TT$ // $\FF$ // $n$ :
        ;;
        n :in: $\mathbb{Z}$
      \end{bnf}
    \end{center}
  \end{minipage}%
  \begin{minipage}{0.4791\textwidth}
    \begin{center}
      \begin{bnf}[rccll]
        cl : $\dom{Clos}$ ::= $\clos{x}{e}{\sigma}$ :
        ;;
        \sigma : $\dom{Env}$ ::= $[\Overline{x \mapsto v}]$ :
        ;;
        \omega : $\dom{Buffer}$ ::= $[\Overline{v}]$
      \end{bnf}
    \end{center}
  \end{minipage}
\end{center}
Base values are standard---constants~$k$ include the unit value~$\<\>$, Booleans~$\TT$ and~$\FF$, and integers~$n$.
To model the base layer using big-step semantics \citep{Kah87Natural} in~\cref{subsubsec:body}, bound variables are recorded in an environment~$\sigma = [\Overline{x \mapsto v}]$. 
A function evaluates into a closure~$cl = \clos{x}{e}{\sigma}$. 
A buffer~$\omega = [\Overline{v}]$ holds a list of printed values.

\subsubsection{React Machinery}\label{subsubsec:machinery}
\paragraph{Extended Values}
To model the React machinery, we extend base values~$\dom{Val}$:
  \begin{center}
    \begin{minipage}{0.4791\textwidth}
      \begin{center}
        \begin{bnf}[rccll]
          v : $\dom{Val}$ ::=  $\dotsb$ // $C$  // $cs$ // $[\Overline{s}]$ // $\<\ell, p\>$ :
          ;;
          cs : $\dom{ComSpec}$ ::= $\<C, v\>$ :
          ;;
          s : $\dom{ViewSpec}$ ::= $k$ // $cl$ // $cs$ // $[\Overline{s}]$ :
        \end{bnf}
      \end{center}
    \end{minipage}%
    \begin{minipage}{0.4791\textwidth}
      \begin{center}
        \begin{bnf}[rccll]
          \delta : $\dom{DefTable}$ ::= $[\Overline[\ell]{C \mapsto \deftabent{x}{e}}]$ :
          ;;
          \ell :in: $\mathbb{N}$ :
          ;;
          p :in: $\dom{Path} = \mathbb{N}$ :
        \end{bnf}
      \end{center}
    \end{minipage}
  \end{center}
A component name~$C$ is a value as-is.
All component names are determined statically and their definitions are looked up using a \emph{definition table~$\delta = [\Overline[\ell]{C \mapsto \deftabent{x}{e}}]$}.
Thus (mutually) recursive components are permitted in \Rtrace, and we do not model legacy higher-order components \cite{Met25Higher}.
A \emph{view spec} is either a constant value~$k$, a closure~$cl$, a component spec~$cs$, or an array view spec~$[\Overline{s}]$.
A constant and a closure view spec represent leaf views, whereas a component and an array view spec represent composite views.
A constant view spec is simply realized into a view that shows its string representation.
A closure view spec, when realized into a view, represents an event handler---this models UI elements such as a button, an input field, a checkbox, etc.
A component spec~$cs = \<C, v\>$ is simply a pair of a component name and its argument.
This represents a specification of a view that may be realized into a view hierarchy.
An array view spec~$[\Overline{s}]$ is an array of view specs~$s$, representing structured child views.

A setter closure~$\<\ell, p\>$, which is the semantic value of a setter function, is a pair of a label~$\ell$ from the originating \reacttt{useState} Hook and a \emph{path}~$p$ to the corresponding view in the view tree.
A path~$p$ is a unique natural number used to identify each view in the view tree.
We shall take a look at the view tree structure~$m$ in a moment.

\paragraph{Tree Memory}
We introduce \emph{tree memory}~$m$ to model the view hierarchy:
\begin{center}
  \begin{bnf}[rccll]
    m : $\dom{TreeMem}$ ::= $[\Overline{p \mapsto \pi}]$ :
    ;;
    \pi : $\dom{View}$ ::= $\view{cs}{\{\Overline{d}\}}{\rho}{q}{t}$ :
    ;;
    d : $\dom{Decision}$ ::= \dec{Check} // \dec{Effect} :
  \end{bnf}
  \vspace*{-0.2\baselineskip}
  \begin{center}
    \begin{minipage}{0.4791\textwidth}
      \begin{center}
        \begin{bnf}[rccll]
          t : $\dom{Tree}$ ::= $k$ // $cl$ // $p$ // $[\Overline{t}]$ :
          ;;
          \rho : $\dom{SttStore}$ ::= $[\Overline[\ell]{\ell \mapsto \sttstent{v}{q}}]$ :
        \end{bnf}
      \end{center}
    \end{minipage}%
    \begin{minipage}{0.4791\textwidth}
      \begin{center}
        \begin{bnf}[rccll]
          q : $\dom{JobQ}$ ::= $[\Overline[\ell]{cl}]$ :
          ;;
          \Sigma : $\dom{Context}$ ::= $m$ // $\pi$ :
        \end{bnf}
      \end{center}
    \end{minipage}
  \end{center}
\end{center}
The entire view hierarchy is stored in a tree memory~$m$ which maps each path~$p$ to a view~$\pi$.
A view is a realized component that is mounted or to-be-mounted into a tree memory.
Each mounted view~$\pi$ is assigned a unique path~$p$ where $p$~is a valid path of a tree memory~$m$, i.e., $p \in \domain m$.
It consists of a component spec~$cs$, a decision~$d$, an effect queue~$q$, and a child tree~$t$.
When a component body is evaluated, the runtime keeps track of its decision set---\dec{Check} for checking the component for retry or re-render, and \dec{Effect} for running Effects after this render.
Decisions \dec{Check} and \dec{Effect} have been explained informally in~\cref{subsubsec:infloop,subsubsec:topset}, and all components first begin with an empty decision set.
A state store~$\rho$ for each view is used to store the current state value~$v$ and the queued updates~$q$ from setter function from the \reacttt{useState} Hook.
These are keyed using the corresponding label~$\ell$ of \reacttt{useState}.
An effect queue~$q$ is a queue of Effects from the \reacttt{useEffect} Hook.
A child~$t$ is either a terminal view~$k$, an event handler~$cl$, a path~$p$ to another view, or an array of children~$[\Overline{t}]$.
We also introduce a context~$\Sigma$, which can be either a tree memory~$m$ or a view~$\pi$.
This distinction exists because one evaluation rule (\Rule{AppSetNormal}) requires the whole tree memory~$m$, while the others only require the local view~$\pi$ in their context.
The whole-memory context is used to describe a component updating another, e.g., a button in a child component updating its parent's state.

\begin{figure}[t]
  \begin{center}
    \begin{minipage}{.27\textwidth}
      \begin{reactcode}
  let Bin n =
    if n = 0 then () else
    [Bin(n-1), Bin(n-1)];;
  Bin 2
      \end{reactcode}
    \end{minipage}%
    \begin{minipage}{.73\textwidth}
      \includegraphics[width=\linewidth]{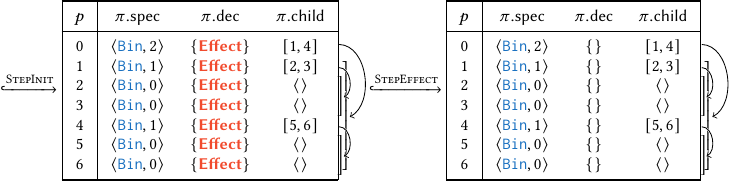}
    \end{minipage}
  \end{center}
  \caption{Visualization of a tree memory and render step transitions for a simple recursive component.}\label{fig:bin}
\end{figure}

To give a visual intuition of how tree memory represents view hierarchy, \cref{fig:bin} shows a tree memory for a simple recursive component~\reacttt{Bin} that creates a complete binary tree of height~2.
The root tree is assigned path~$p = 0$, and each subtree receives a unique path in the order of initialization (\cref{fig:init},~\cref{subsubsec:body}).
The view hierarchy can be reconstructed by traversing the \field{child} field of each view in the tree memory.
The render step transitions \Rule{StepInit} and \Rule{StepEffect} along with the decisions shown in \cref{fig:bin} are explained in detail in~\cref{subsubsec:lifecycle}.

\paragraph{Phase}
While evaluating an expression, we maintain a phase~$\phi$:
\begin{center}
  \begin{bnf}[rccll]
    \phi : $\dom{Phase}$ ::= \phase{Init} // \phase{Succ} // \phase{Normal} :
  \end{bnf}
\end{center}
\phase{Init}\footnote{We use \phase{blue sans-serif} to emphasize a phase.} phase is used when evaluating a component for the first time to be mounted into a tree memory.
At this phase, all states are evaluated from the initial expressions of \reacttt{useState}s.
Successive calls to a component function are done in \phase{Succ} phase, where states are retrieved from a state store.
The main expression, Effects, and event handlers are evaluated in \phase{Normal} phase.

\paragraph{Mode}
We model each (re-)render as a transition of global states, and we introduce a mode~$\mu$ to keep track of the state of the \Rtrace{} engine:
\begin{center}
  \begin{bnf}[rccll]
    \mu : $\dom{Mode}$ ::= $\rendermode$ \textrm{(rendered)} // $\checkmode$ \textrm{(check)} // $\eloopmode$ \textrm{(event loop)} :
  \end{bnf}
\end{center}
Rendered mode~$\rendermode$ represents a state where the UI has been rendered on the screen and is waiting for the queued Effects to run.
Check mode~$\checkmode$ represents a state where the runtime will check for a re-render.
When a re-render is not required, we enter event loop mode~$\eloopmode$ and wait for an input.

\subsection{Operational Semantics}\label{subsec:sem}
All semantic functions\footnote{Semantic functions are typeset using an \sem{italicized sans-serif} font.} and evaluation relations, along with their dependencies and required contexts, are summarized in \cref{fig:dependencies}.
\begin{figure}[t]
  \centering
  \includegraphics[width=.85\textwidth]{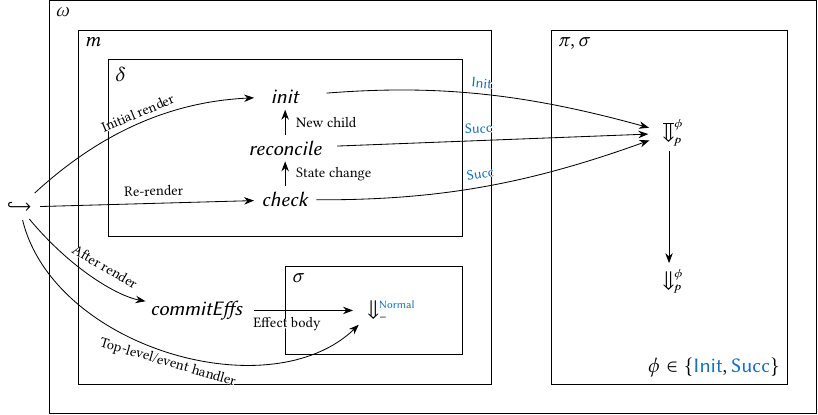}
  \Description{A diagram of dependencies between semantic functions.}
  \caption{Semantic function dependencies.}\label{fig:dependencies}
\end{figure}
We give a brief overview of the semantic functions and relations before we formally introduce them in~\cref{subsubsec:body,subsubsec:commiteffs,subsubsec:reconciliation}:
\begin{itemize}
  \item $\<e, \delta\> \text{ or }\<t, m, \omega, \delta, \mu\> \smallstep \<t, m', \omega', \delta, \mu'\>$ is a render step transition (\cref{fig:step}).
    Read this as ``Initially, under definition table~$\delta$, main expression~$e$ transitions to root tree~$t$, tree memory~$m'$, and output buffer~$\omega'$, entering mode~$\mu'$,'' or ``Given root tree~$t$ under definition table~$\delta$, tree memory~$m$ and output buffer~$\omega$ in mode~$\mu$ transition to~$m'$ and $\omega'$, entering mode~$\mu'$.''

  \item $\Step<\phi><p>[\Sigma, \sigma]{e}{v, \Sigma', \omega}$ is an evaluation of an expression (\cref{fig:eval}).
    Read this as ``In phase~$\phi$ at path~$p$, expression~$e$ under environment~$\sigma$ and context~$\Sigma$ evaluates to value~$v$, producing a modified context~$\Sigma'$ and output~$\omega$.''

  \item $\Step*<\phi><p>[\pi, \sigma]{e}{s, \pi', \omega}$ is a possibly retrying evaluation of a component body until it becomes idle (\cref{fig:evalmult}).
    Read this as ``In phase~$\phi$ at path~$p$, component body~$e$ of view~$\pi$ under environment~$\sigma$ eventually evaluates to view spec~$s$, producing a modified view~$\pi'$ and output~$\omega$.''

  \item $m, \delta \vdash \sem{init}(s) = \<t, m', \omega\>$ is a rendering of a view spec in an initial render (\cref{fig:init}).
    Read this as ``Under definition table~$\delta$, view spec~$s$ initially renders into tree~$t$, modifying tree memory from~$m$ to~$m'$ and printing~$\omega$.''

  \item $m \vdash \sem{commitEffs}(t) = \<m', \omega\>$ is committing queued Effects (\cref{fig:commiteffs}).
    Read this as ``Committing Effects of tree~$t$ modifies tree memory from~$m$ to~$m'$ and prints~$\omega$.''

  \item $m, \delta \vdash \sem{check}(t) = \<\mu, m', \omega\>$ is checking and updating a tree for a re-render (\cref{fig:check}).
    Read this as ``Under definition table~$\delta$, tree~$t$ re-renders with modified tree memory~$m'$ from~$m$ and prints~$\omega$'' when~$\mu = \rendermode$, and ``Under definition table~$\delta$, tree~$t$ does not re-render but modifies tree memory from~$m$ to~$m'$ and prints~$\omega$'' when~$\mu = \eloopmode$; $\mu$ cannot be $\checkmode$.

  \item $m, \delta \vdash \sem{reconcile}(t, s) = \<t', m', \omega\>$ is reconciling a tree with a possibly updated view spec (\cref{fig:reconcile}).
    Read this as ``Under definition table~$\delta$, tree~$t$ is reconciled as~$t'$ with respect to view spec~$s$, modifying tree memory from~$m$ to~$m'$ and printing~$\omega$.''
\end{itemize}

\begin{nota}
We explain notation used in the evaluation rules.
When a rule has premises~$J_i$ for~$l \le i \le u$, we write~$(J_i)_{i = l}^u$.
We use~$\concat$ for list concatenation, e.g., $x \concat x'$, and~$\Concat_{i=l}^u x_i$ for concatenating~$x_l, \dotsc, x_u$.
For bindings~$X$, $X[x \mapsto v]$ updates~$X$ with~$x \mapsto v$, and~$X[x]$ retrieves the value~$v$ when $X$~has a binding~$x \mapsto v$.
For a structure~$X$ with a field~\field{f}, $X\{\fv{f}{v}\}$ updates field~\field{f} to~$v$, and~$X.\field{f}$ accesses field~\field{f}.
We use meta-level let-bindings~``let $x = v$ in $\dotsc$'' inside meta-level structure or list updates when convenient.
Meta-level conditional~$\scond{b}{x}{y}$ chooses~$x$ when~$b$ holds and chooses~$y$ otherwise.
\end{nota}

\subsubsection{Render Loop}\label{subsubsec:lifecycle}
We model the React event loop with render step transitions in \cref{fig:step}.
\begin{figure}[t]
  \centering
  \input{semantics/step}
  \Description{The rules modeling render step transitions.}
  \caption{Render step transitions.}\label{fig:step}
\end{figure}

Given a program~$P$, the loop begins with an initial state~$\<e, \delta\>$, where $e$~is the main expression of~$P$ and $\delta$ is the definition table of all components defined in~$P$.
Note that the definition table~$\delta$ and the root tree~$t$ (once initialized) remain unchanged during transitions.

First, we evaluate the main expression~$e$ into a view spec~$s$ and initialize it into a tree~$t$ with an updated tree memory~$m$; we enter the rendered mode~$\rendermode$ (\Rule{StepInit}).
We call this a rendered mode because the root tree~$t$ and all of its descendants as described in~$m$ have been \emph{rendered on screen.}
Then we commit all the queued Effects of~$t$ and its descendants, after which we enter the check mode~$\checkmode$ (\Rule{StepEffect}).
In check mode~$\checkmode$, we check for any state update under~$t$.
If there is an update, we render again and enter the rendered mode~$\rendermode$; otherwise, we enter the event loop mode~$\eloopmode$ (\Rule{StepCheck}).
In event loop mode~$\eloopmode$, we wait for user input to be handled by an event handler.
If input occurs, we evaluate the corresponding handler body, which then requires a state update check (\Rule{StepEvent}).

Returning to the simple example shown in \cref{fig:bin}, all the views in the tree memory contain \dec{Effect} decisions after the \Rule{StepInit} transition, as they are freshly mounted and React needs to check for any defined Effects (of which there are none in the example).
All the \dec{Effect} decisions are then cleared after the \Rule{StepEffect} transition.
No state updates have been made---there are no \dec{Check} decisions---and \Rule{StepCheck} transitions to the same tree memory, this time in event loop mode~$\eloopmode$.

To elaborate on how the transition \Rule{StepInit} evaluates the main expression and initializes the evaluated view spec, we turn our attention to the evaluation rules for expressions~(\cref{subsubsec:eval}) and the initialization rules (\cref{subsubsec:body}).

\subsubsection{Evaluating an Expression}\label{subsubsec:eval}
\paragraph{Single Evaluation}
We present how an expression is evaluated in a big-step semantics in \cref{fig:eval}.
The rules relevant to the render logic are included here, which includes \reacttt{useState} and the \reacttt{useEffect} Hooks, component application, and state updates.
Standard rules related to the base logic are not presented---\Rule{AppFunc} and \Rule{Print} are included for clarity.
The complete operational semantics of expressions is available in~\S A.2. 
\begin{figure}[t]
  \centering
  \hfill\fbox{$\Step<\phi><p>[\Sigma, \sigma]{e}{v, \Sigma', \omega}$}
  \begin{mathpar}
    \input{semantics/eval-react}
  \end{mathpar}
  \Description{An excerpt of the rules modeling an expression evaluation.}
  \caption{Evaluation of an expression (an excerpt).}\label{fig:eval}
\end{figure}

An environment~$\sigma$, context~$\Sigma$, phase~$\phi$, and path~$p$ are required for evaluating an expression~$e$.
A path is not available (written as $-$ in \cref{fig:eval}) in the \phase{Normal} phase.
Evaluation rules that are only applicable in the context of a component body evaluation (\Rule{AppSetComp}, \Rule{SttBind}, \Rule{SttReBind}, and \Rule{Eff}), i.e., in the \phase{Init} and \phase{Succ} phase, require the context~$\Sigma$ to be a view~$\pi$.
A noteworthy case is \Rule{AppSetNormal}, the only rule applicable in the \phase{Normal} phase only.
All other rules permit both~$\pi$ and~$m$ as a context~$\Sigma$.
The output buffer~$\omega$ stores the printed values sequentially (\Rule{Print}).

It is important to note that component definitions are not required during evaluation of an expression.
Hence, a component application is nothing more than packing the evaluated component name~$C$ and the argument~$v$ in a pair (\Rule{AppCom}).
The actual function bound to the component name is neither invoked nor required at this point.

The \reacttt{useState} Hook behaves differently in the \phase{Init} and \phase{Succ} phases.
In the \phase{Init} phase, the initial value~$v_1$ is evaluated from the argument~$e_1$, which is then stored in the view's state store \field{sttst} at label~$\ell$ of \reacttt{useState}, along with an empty state update queue \field{sttq} (\Rule{SttBind}).
In addition, the state variable~$x$ and the setter function~$x_{\texttt{set}}$ are bound to the initial value~$v_1$ and the setter closure, which is a pair~$\<\ell, p\>$, so that the value of the state can be read and set.

In the \phase{Succ} phase, all the queued updates in \field{sttq} are processed while evaluating the corresponding \reacttt{useState} Hook (\Rule{SttReBind}).
After evaluating all the updates, we compare the final state value~$v_n$ with the initial value~$v_0$.
If they are equivalent ($v_n \equiv v_0$), we simply keep the previous decision \field{dec}.
Otherwise ($v_n \nequiv v_0$), the component adds the \dec{Effect} decision to its decision so that it runs its Effects after rendering.
The state update queue \field{sttq} is flushed and the state variable and the setter function are bound appropriately.

State updates via setter function applications behave differently when evaluating a component ($\phi \in \{\phase{Init}, \phase{Succ}\}$) compared to when evaluating Effects or event handlers ($\phi = \phase{Normal}$).
During component body evaluation, calling a setter function of another component is not allowed, hence the path~$p$ in the setter closure must match with the context (\Rule{AppSetComp}).
The official React runtime logs an error message when a component evaluation invokes a setter function of another component.
Setting the component's state during the evaluation of its body adds the \dec{Check} decision and queues the provided update closure~$cl$.
The \dec{Check} decision triggers a re-evaluation of the component---this behavior is formalized in \cref{fig:evalmult}.

State updates in the \phase{Normal} phase lift the restriction of ``inter-component'' updates.
The setter closure~$\<\ell, p\>$ is used to queue the update closure~$cl$ in the correct tree at path~$p$ and the correct label~$\ell$ in the state queue \field{sttq} (\Rule{AppSetNormal}).
The \dec{Check} decision is added to the view's decision, to mark that a state update has been queued.
The marked views are later processed in batches~(\cref{subsubsec:reconciliation}).

\paragraph{Retrying Evaluation}
The evaluation of a component body starts with a \emph{retrying evaluation}, where the body is repeatedly evaluated until the \(\dec{Check}\) decision is no longer present.  
Before each evaluation, the Effect queue is cleared so that Effects are re-collected during execution.  
The rationale is that re-evaluated view specs are discarded except for the last one, and Effects associated with the discarded view specs must also be discarded.
In addition, the \(\dec{Check}\) decision is removed beforehand to determine whether the evaluation triggers another \(\dec{Check}\).  
The component body is then evaluated, and evaluation stops when \(\dec{Check}\) is no longer present (\Rule{EvalOnce}).  
If the evaluation results in another \(\dec{Check}\) decision, re-evaluation is performed in the \phase{Succ} phase (\Rule{EvalMult}).

\begin{figure}[t]
  \centering
  \input{semantics/evalmult}
  \Description{The rules modeling a retrying component body evaluation.}
  \caption{Retrying evaluation of a component body.}\label{fig:evalmult}
\end{figure}

Note that a component may indefinitely decide to \dec{Check}, which will lead to an infinite retrial issue described in \cref{subsubsec:topset}.
React runtime raises an exception after 25~retrials.

\subsubsection{Initial Render}\label{subsubsec:body}
When a view is initially rendered, it is initialized from a view spec (\cref{fig:init}).
A constant (\Rule{InitConst}), closure (\Rule{InitClos}), and array view spec initialization (\Rule{InitArray}) are straightforward.
\begin{figure}[t]
  \centering
  \input{semantics/init}
  \Description{The rules modeling a view spec initialization.}
  \caption{Initialization of a view spec.}\label{fig:init}
\end{figure}

To initialize a component spec~$\<C, v\>$, a fresh path~$p$ with respect to the tree memory~$m$ is generated, and the definition table~$\delta$ is used to look up the component definition for name~$C$ (\Rule{InitCom}).
First, an empty view with the component spec, an empty decision set, and a (placeholder)~unit child is created.
This view is used as a context to evaluate the component body in the \phase{Init} phase with the parameter~$x$ bound to~$v$, which produces a view spec~$s$ and an updated view~$\pi$.
After mounting the updated~$\pi$ at path~$p$ in tree memory~$m$, view spec~$s$ is recursively initialized to get the child tree~$t$.
Finally, the \dec{Effect} decision is added to the view to ensure the queued Effects are run after the initial render, and the \field{child} field is set to~$t$.
Note that the initialization of the child tree does not modify the parent view~$\pi$.

\subsubsection{Committing Effects After Render}\label{subsubsec:commiteffs}
After a view has been rendered, the queued Effects must be executed.
The rules in \cref{fig:commiteffs} describe this process.
For a constant (\Rule{CommitEffsConst}) and a closure (\Rule{CommitEffsClos}), there are no Effects to commit, so the tree memory remains unchanged.
For an array of trees (\Rule{CommitEffsArray}), we recursively commit the Effects of each child tree.
\begin{figure}[t]
  \centering
  \input{semantics/commitEffs}
  \Description{The rules modeling committing Effects.}
  \caption{Committing Effects.}\label{fig:commiteffs}
\end{figure}

For a path to a view, Effects are committed only when the view has decided to.
If the view has not decided to commit \dec{Effect}s, only the Effects of its child are executed, and the view's Effects are skipped (\Rule{CommitEffsPathIdle}).
If the view's decision includes \dec{Effect}, the Effects of its child are committed first, and then the view's Effects are executed in order (\Rule{CommitEffsPath}).
Each Effect is evaluated in the \phase{Normal} phase, allowing it to update any component's states in the tree.

Effects execution follows a post-order traversal---the child's Effects are executed first, followed by those of its parent.
Additionally, Effects are executed following their original syntactic order.

\subsubsection{Checking, Re-Render, and Reconciliation}\label{subsubsec:reconciliation}
\paragraph{Checking for Re-Render}
In check mode~$\checkmode$, which is after committing Effects or handling an event, the runtime checks and re-renders a tree if needed.
This is carried out with semantic function~\sem{check} shown in \cref{fig:check}.
When a tree or any of its descendants requires a re-render, \sem{check} performs the re-render and returns either of~$\rendermode$ or~$\eloopmode$.
It also returns the modified tree memory.
\begin{figure}[t]
  \centering
  \input{semantics/check}
  \Description{The rules modeling checking a tree for re-render.}
  \caption{Checking a tree for re-render.}\label{fig:check}
\end{figure}

For a constant value (\Rule{CheckConst}) and a closure (\Rule{CheckClos}), no render is needed and \sem{check} always returns~$\eloopmode$ without any modification to the tree memory.
For a tree array, we check each of its trees recursively in sequence (\Rule{CheckArray}).
Note the~$\sqcup$ operation in \Rule{CheckArray}, defined as a commutative operator between~$\{\eloopmode, \rendermode\}$ where~$\eloopmode \sqcup \eloopmode = \eloopmode$ and $\rendermode \sqcup \rendermode = \rendermode \sqcup \eloopmode = \rendermode$.
This is because returning~$\rendermode$ indicates that any descendant has been re-rendered.

The interesting cases occur with path references to views.
An idle view that does not need \dec{Check}ing simply skips and recurse on its child (\Rule{CheckIdle}).
If a view's decision includes \dec{Check}, we need to re-evaluate the component body.

When the re-evaluation of a view with \dec{Check} decision results in a decision without \dec{Effect}, it means that the view eventually settled with the same state as before, and no re-render happens (\Rule{CheckNoEffect}).
When the re-evaluation decides to commit \dec{Effect}s, we need to reconcile the child tree against the new view spec (\Rule{CheckEffect}).
In this case, we return $\rendermode$ to indicate that the view's state has changed.
Note that the re-evaluation premises of \Rule{CheckNoEffect} and \Rule{CheckEffect} do not modify the child~$\pi.\field{child}$, as an evaluation does not touch the \field{child} field (\cref{subsubsec:eval}).

\paragraph{Reconciliation}
\begin{figure}[tb]
  \centering
  \input{semantics/reconcile}
  \Description{The rules modeling reconciling a tree with a view spec.}
  \caption{Reconciliation of a tree with a view spec.}\label{fig:reconcile}
\end{figure}
When some states of a view update, it needs to be \emph{reconciled} with the updated view spec, which is described in \cref{fig:reconcile}.
For reconciling an array tree against an array view spec, we reconcile each child tree with the corresponding view spec (\Rule{ReconcileArray}).

For a path to a component view, if the component name is the same as before, we re-evaluate the component body in the \phase{Succ} phase and then reconcile the old child with the new view spec (\Rule{ReconcileComEffect}).
If the component name has changed, we re-initialize the view spec as this is a completely new component compared to before (\Rule{ReconcileComNew}).

For all other cases, such as when a constant, closure, or array tree transitions to a different type, we re-initialize the view spec from scratch (\Rule{ReconcileOther}).

This reconciliation process allows React to efficiently update the UI when state changes occur, preserving existing (virtual) DOM nodes where possible, and only rebuilding the parts that have actually changed \citep{Cac16React}.

\subsection{An Illustrative Example}\label{subsec:example}
\begin{wrapfigure}[7]{r}{0.37\linewidth}
  \vspace*{-\intextsep}
  \begin{reactcode}
let Demo x =
  let (s, setS) = useState x in
  let f = fun s -> s + 1 in
  if s = 0 then setS f;
  useEffect (if s = 1 then setS f);
  if s <= 1 then () else
    button (fun _ -> setS f);;
Demo 0
  \end{reactcode}
\end{wrapfigure}
To illustrate the operational semantics of Hooks in action, we walk through the execution of the \reacttt{Demo} component on the right.
The example incorporates both the \reacttt{useState} and \reacttt{useEffect} Hooks, showcasing the unnecessary re-rendering issue (\cref{subsec:unnec}) and the top-level setter issue (\cref{subsubsec:topset}).
\reacttt{Demo} also demonstrates reconciliation by updating its child view from~\reacttt{()} to \reacttt{button} (or a closure) after the re-render.

The render step transitions for |Demo 0| is illustrated in the diagram below.
The root path is~$p_0$ and each view in each step is indexed, e.g.,~$\pi_0, \dots, \pi_4$.
For brevity, the state store has been flattened as there is only a single state in \reacttt{Demo} and closures are abbreviated.
In addition, the intermediate state that triggers a retry is listed below the initial configuration~$\<\texttt{Demo 0}, [\texttt{Demo} \mapsto \dots]\>$.
\begin{center}
  \smaller[2]
  \setlength{\tabcolsep}{2.5pt}

  \hphantom{$\overset{\Rule{StepCheck}}{\smallstep}$}
  \begin{minipage}[t]{0.22\linewidth}
    \centering
    $\<\texttt{Demo 0}, [\texttt{Demo} \mapsto \dots]\>$ \\[1ex]
    \smaller[1]
    \begin{tabular}{|rl|}
      \hline
      \field{dec} & $\{\dec{Check}\}$ \\
      \field{val} & $0$ \\
      \field{sttq} & $[\<\lfun{\texttt{s}}{\texttt{s+1}}, \_\>]$ \\
      \field{effq} & $[\<\texttt{if s=1.\!.\!.}, \_\>]$ \\
      \field{child} & $\<\>$ \\
      \hline
    \end{tabular}
    (During \Rule{EvalMult})
  \end{minipage}
  $\overset{\Rule{StepInit}}{\smallstep}$
  \begin{minipage}[t]{0.22\linewidth}
    \centering
    $\<p_0, [p_0 \mapsto \pi_0], \delta, \rendermode\>$ \\[1ex]
    \smaller[1]
    \begin{tabular}{|rl|}
      \hline
      \field{dec} & $\{\dec{Effect}\}$ \\
      \field{val} & $1$ \\
      \field{sttq} & $[]$ \\
      \field{effq} & $[\<\texttt{if s=1.\!.\!.}, \_\>]$ \\
      \field{child} & $\<\>$ \\
      \hline
    \end{tabular}
  \end{minipage}
  $\overset{\Rule{StepEffect}}{\smallstep}$
  \begin{minipage}[t]{0.2\linewidth}
    \centering
    $\<p_0, [p_0 \mapsto \pi_1], \delta, \checkmode\>$ \\[1ex]
    \smaller[1]
    \begin{tabular}{|rl|}
      \hline
      \field{dec} & $\{\dec{Check}\}$ \\
      \field{val} & $1$ \\
      \field{sttq} & $[\<\lfun{\texttt{s}}{\texttt{s+1}}, \_\>]$ \\
      \field{effq} & $[]$ \\
      \field{child} & $\<\>$ \\
      \hline
    \end{tabular}
  \end{minipage}
  \\[1.5ex]
  $\overset{\Rule{StepCheck}}{\smallstep}$
  \begin{minipage}[t]{0.2\linewidth}
    \centering
    $\<p_0, [p_0 \mapsto \pi_2], \delta, \rendermode\>$ \\[1ex]
    \smaller[1]
    \begin{tabular}{|rl|}
      \hline
      \field{dec} & $\{\dec{Effect}\}$ \\
      \field{val} & $2$ \\
      \field{sttq} & $[]$ \\
      \field{effq} & $[\<\texttt{if s=1.\!.\!.}, \_\>]$ \\
      \field{child} & $\<\lfun{\_}{\texttt{setS f}}, \_\>$ \\
      \hline
    \end{tabular}
  \end{minipage}
  $\overset{\Rule{StepEffect}}{\smallstep}$
  \begin{minipage}[t]{0.2\linewidth}
    \centering
    $\<p_0, [p_0 \mapsto \pi_3], \delta, \checkmode\>$ \\[1ex]
    \smaller[1]
    \begin{tabular}{|rl|}
      \hline
      \field{dec} & $\{\}$ \\
      \field{val} & $2$ \\
      \field{sttq} & $[]$ \\
      \field{effq} & $[]$ \\
      \field{child} & $\<\lfun{\_}{\texttt{setS f}}, \_\>$ \\
      \hline
    \end{tabular}
  \end{minipage}
  $\overset{\Rule{StepCheck}}{\smallstep}$
  \begin{minipage}[t]{0.2\linewidth}
    \centering
    $\<p_0, [p_0 \mapsto \pi_4], \delta, \eloopmode\>$ \\[1ex]
    \smaller[1]
    \begin{tabular}{|rl|}
      \hline
      \field{dec} & $\{\}$ \\
      \field{val} & $2$ \\
      \field{sttq} & $[]$ \\
      \field{effq} & $[]$ \\
      \field{child} & $\<\lfun{\_}{\texttt{setS f}}, \_\>$ \\
      \hline
    \end{tabular}
  \end{minipage}
\end{center}

Let's follow the execution of component \reacttt{Demo} step by step:
\begin{enumerate}[start=0]
  \item The main expression is |Demo 0|, and the definition table~$\delta$ contains a single entry of \reacttt{Demo}.

  \item |Demo 0| evaluates into a component spec, which is initialized (\Rule{StepInit}).
    During initialization, the component body is evaluated in the \phase{Init} phase (\Rule{EvalMult}):
    \begin{enumerate}
      \item \reacttt{useState} initializes state variable~\reacttt{s} to~0 and setter function \reacttt{setS} (\Rule{SttBind}).
      \item Since $\text{\reacttt{s}} = 0$, a top-level call to \reacttt{setS} is made and \dec{Check} decision is on (\Rule{AppSetComp}).
      \item \reacttt{useEffect} queues the Effect body (\Rule{Eff}).
      \item Since $\text{\reacttt{s}} \le 1$, a unit is returned.
    \end{enumerate}
    Since \dec{Check} is on, \reacttt{Demo}'s body is re-evaluated in the \phase{Succ} phase with an empty decision and an empty effect queue (\Rule{EvalOnce}):
    \begin{enumerate}
      \item \reacttt{useState} processes the queued update, changing state~\reacttt{s} to~1 (\Rule{SttReBind}).
        Since the state is different from the previous~|0|, \dec{Effect} decision is made.
      \item The rest is similar to the above, so we summarize: the top-level call to \reacttt{setS} is avoided, the Effect body is queued (\Rule{Eff}), and still $\text{\reacttt{s}} = 1 \le 1$ and hence a unit is returned.
    \end{enumerate}

  \item Now we are in rendered mode~$\rendermode$, and we commit Effects (\Rule{StepEffect}).
    \begin{enumerate}
      \item The view decided to run \dec{Effect}s, so the queued Effect is evaluated in the \phase{Normal} phase (\Rule{CommitEffsPath}). Note that there is no child view's Effect to run (\Rule{CommitEffsConst}).
      \item In the Effect, $\text{\reacttt{s}} = 1$ and hence \reacttt{setS} is called.
        \dec{Check} decision is added and the state update closure is queued (\Rule{AppSetNormal}).
      \item \dec{Effect} decision is removed and the decision set is now $\{\dec{Check}\}$ (\Rule{CommitEffsPath}).
    \end{enumerate}

  \item Now we are in check mode~$\checkmode$, and we check for a re-render (\Rule{StepCheck}).
    \begin{enumerate}
      \item The component body evaluates again (\Rule{EvalOnce}): the state update to~2 turns on \dec{Effect} (\Rule{SttReBind}), the Effect is queued (\Rule{Eff}), and the button event handler is returned.
      \item The previous child~$\<\>$ is reconciled with the closure view spec~$\<\lfun{\_}{\texttt{setS f}}, \_\>$ (\Rule{Check\-Effect}).
      \item The closure view spec is initialized as it is of different type with~$\<\>$ (\Rule{ReconcileOther}).
    \end{enumerate}

  \item We are again in rendered mode~$\rendermode$, so we commit Effects (\Rule{StepEffect}).
    \begin{enumerate}
      \item Again, the view has $\dec{Effect}$ on and we commit the Effect (\Rule{CommitEffsPath}), but \reacttt{setS} is not called this time.
        Then we turn off $\dec{Effect}$, leaving the decision set empty (\Rule{CommitEffsPath}).
    \end{enumerate}

  \item We are in check mode~$\checkmode$ to check for a re-render (\Rule{StepCheck}), but this time the view is idle and we are done (\Rule{CheckIdle}).
\end{enumerate}

\section{The \Rtrace{} Interpreter and the Visualizer}\label{sec:impl}
\begin{figure}[t]
  \centering
  \includegraphics[width=\textwidth]{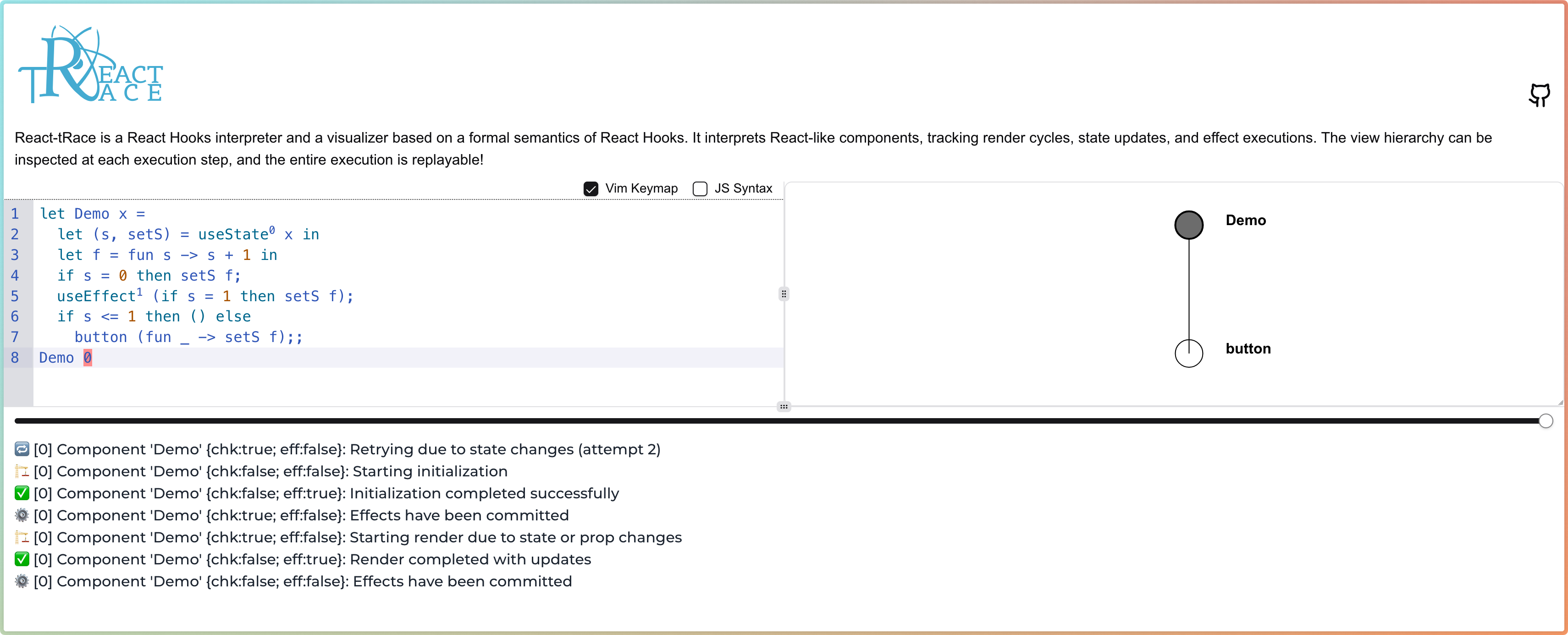}
  \Description{A screenshot of the \Rtrace{} interpreter and visualizer web UI is shown.}
  \caption{The \Rtrace{} interpreter and visualizer interface, showing the illustrative example from~\cref{subsec:example}.}\label{fig:vis}
\end{figure}

We have implemented a definitional interpreter and a visualization tool (\cref{fig:vis}) to help developers understand the behavior of React Hooks, based on the semantics of React Hooks~(\cref{sec:sem}).

The visualizer enables React programmers to examine their programs in an interactive manner:
\begin{itemize}
  \item The view hierarchy can be inspected by visualizing the tree memory.
  \item The execution of the program is explained, automating the step-by-step reasoning in~\cref{subsec:example}.
  \item The entire execution is replayable using a slider, allowing thorough examination of re-rendering bugs described in~\cref{sec:pit}.
\end{itemize}

While our visualizer is fully functional and readily usable, we are actively developing additional features to enhance its capabilities.
Planned or in-progress enhancements include more comprehensive state and trace visualization, experimental JS translation to bridge the gap between real-world React applications, and a direct preview of the rendered UI.

\paragraph{Implementation} The \Rtrace{} interpreter is a faithful implementation of our semantics written in OCaml~5, with the visualization front-end implemented using React with ReScript (\href{https://rescript-lang.org/}{rescript-lang.org}) and TypeScript (\href{https://www.typescriptlang.org/}{typescriptlang.org}).
We used the Js\_of\_ocaml compiler \citep{VouBal14Bytecode} to emit JS code that interfaces with our front-end.


The source code is available at \href{https://github.com/Zeta611/react-trace}{\gh{Zeta611/react-trace}}, and the visualizer interface can be accessed online at \href{https://react-trace.vercel.app/}{\nolinkurl{react-trace.vercel.app}}.

\section{Conformance with React}\label{sec:rules}
\Rtrace{} is a faithful model of Hooks backed by both theoretical and empirical evaluation:
\begin{enumerate}
  \item \textbf{Conformance according to the official documentation:}
    Although React does not come with a formal definition, we can extract key properties of React based on the documentation.
    We show that \Rtrace{} respects the key properties of React as documented in~\cref{subsec:key}.
  \item \textbf{Direct testing against the implementation:}
    We present the results of empirical tests comparing \Rtrace{} against (multiple releases of) React's runtime behavior across a range of scenarios, covering all evaluation rules, in~\cref{subsec:tests}.
\end{enumerate}

\subsection{Conformance Theorems on React}\label{subsec:key}
The React documentation provides an informal explanation of the behavioral properties of React applications.  
Most of these properties are explicitly captured in our semantics.  
To demonstrate that our semantics align with the behaviors described in the documentation, we prove two key properties in~\cref{subsubsec:key}:
\begin{description}
  \item[\Cref{thm:reeval}:] When a setter function is called during rendering, React immediately attempts to re-evaluate the component body with the updated state \citep{Met25UseState}.
  \item[\Cref{thm:eff}:] A view's Effects are executed after re-rendering triggered by a state update from itself or one of its ancestors \citep{Met25Render}.
\end{description} 

While there are slight differences between our semantics and the React runtime due to an optimization in React, this does not affect reasoning about most React applications.  
We have formalized this optimization and proved that it preserves output behavior (\cref{thm:simtrans},~\cref{subsubsec:opt}), provided that applications adhere to the constraints specified by React.

\subsubsection{Key Properties of React}\label{subsubsec:key}
The first key property states that when a setter function is called during rendering, React immediately attempts to re-evaluate the component body with the updated state.
The pitfall discussed in~\cref{subsubsec:topset} is due to this property.
This behavior is modeled by allowing multiple evaluations of the component body to occur within a retrying evaluation, which happens iff a setter is applied during the first evaluation.

\begin{thm}[Re-Evaluation by Calling Setter]\label{thm:reeval}
The derivation of \(\Step*[\pi, \sigma]{e}{v', \pi'', \omega'}\) includes multiple evaluations of \(e\)
iff \Rule{AppSetComp} appears in the derivation of the first evaluation:
\[ \Step[\pi\{\fv{dec}{\pi.\field{dec} \setminus \{\dec{Check}\}}, \fv{effq}{[]}\}, \sigma]{e}{v, \pi', \omega} \]
\end{thm}

\begin{proof}
Multiple evaluations of \(e\) occur in a retrying evaluation iff the resulting view~\(\pi'\) from the first evaluation includes the \dec{Check} decision (\Rule{EvalMult}).
Thus, it suffices to show that \(\pi'.\field{dec}\) contains \dec{Check} iff \Rule{AppSetComp} appears in the derivation of the first evaluation.

This follows directly from the semantics: \Rule{AppSetComp} is the only rule that introduces \dec{Check} during evaluation (when \(\phi \in \{\phase{Init}, \phase{Succ}\}\)).
Moreover, once added, no rule removes \dec{Check} from the view's decision.
Therefore, \(\pi'\) includes \dec{Check} iff a setter is applied during the first evaluation.
\end{proof}

The second key property states that Effects are executed when queued state updates modify the state values of a view or one of its ancestors (case~1 of \cref{thm:eff}).
Notably, Effects may still be executed even if the final state values remain unchanged, provided state setters are invoked during the evaluation of the component body (case~2 of \cref{thm:eff}).

\begin{thm}[Effect Execution Condition]\label{thm:eff}
If $\<t,m,\omega,\delta,\checkmode\> \smallstep \<t,m',\omega',\delta,\mu\>$,
then for all paths~$p$ reachable from the root in both~$m$ and~$m'$,
\Rule{CommitEffsPath} with~$p$ appears in the derivation of the next transition~$\<t,m',\omega',\delta,\mu\> \smallstep \<t,m'',\omega'',\delta,\mu'\>$
iff there exists an ancestor~$p'$ of~$p$ such that either
\begin{enumerate}
  \item for some $\ell \in \domain(m'[p'].\field{sttst})$, we have $m[p'].\field{sttst}[\ell].\field{val} \ne m'[p'].\field{sttst}[\ell].\field{val}$; or
  \item the derivation of $\<t,m,\omega,\delta,\checkmode\> \smallstep \<t,m',\omega',\delta,\mu\>$ includes \Rule{AppSetComp} with path~$p'$.
\end{enumerate}
\end{thm}

\begin{proof}
We first show only the views in rendered mode~\(\rendermode\) may carry the \dec{Effect} decision. This follows from the operational rules governing state transitions:
\begin{itemize}
  \item \Rule{StepInit} transitions the initial state to rendered mode~\(\rendermode\).
  \item \Rule{StepEffect} clears all \dec{Effect} decisions from descendants of~$t$ via \(\sem{commitEffs}\).
  \item \Rule{StepCheck} ensures \dec{Effect} is added only when transitioning into rendered mode~\(\rendermode\), by induction on \(\sem{check}\).
  \item \Rule{StepEvent} does not introduce \dec{Effect}.
\end{itemize}

The \Rule{CommitEffsPath} rule with path~$p$ appears in the derivation of the next transition iff the \Rule{StepEffect} rule is applied (i.e., \(\mu = \rendermode\)) and \dec{Effect} is included in $m'[p].\field{dec}$.
Among the views reachable in~$m'$, only those updated by \(\sem{check}\) during the first transition may carry the \dec{Effect} decision.
More precisely, \dec{Effect} is added to~$p$ iff~$p$ is a descendant of some~$p'$ such that \(\sem{check}(p')\) is derived via \Rule{CheckEffect}.
This holds iff the resulting view~\(\pi'\) from the retrying evaluation of~$p'$'s body satisfies \(\dec{Effect} \in \pi'.\field{dec}\).

We now show that \dec{Effect} is added to such a view after the retrying evaluation iff either (a)~the final state differs from the initial state, or (b)~setters are applied during the evaluation.

Case~(a) is immediate: state updates occur only via \Rule{SttReBind}, which adds \dec{Effect} when the new value differs from the old.

In case~(b), applying a setter adds \dec{Check} and triggers re-evaluation via \Rule{EvalMult}.
Resolving this re-evaluation loop requires the subsequent evaluation to yield a different result, which necessitates a different binding via \Rule{SttReBind}, thereby introducing \dec{Effect}.
Hence, any terminating derivation involving \Rule{AppSetComp} includes \dec{Effect} in the resulting decision.

Conversely, suppose \dec{Effect} is added, but the final state values are equal to the initial ones.
Then at least one state must have changed and reverted---that is, some label~$\ell$ was updated from~$v$ to~$v'$, and then back to~$v$---implying multiple applications of \Rule{SttReBind} with the same label.
This only occurs when \Rule{AppSetComp} triggers multiple evaluations of the body.
Therefore, \dec{Effect} is added iff either some state value differs from its original value, or \Rule{AppSetComp} appears in the derivation.
\end{proof}

\subsubsection{Optimization in React}\label{subsubsec:opt}
There is a subtle discrepancy between our semantics and the React runtime due to an optimization:
When a state update callback returns the same value as the current state, React skips re-evaluating the component body.

\begin{wrapfigure}[7]{l}{0.377\linewidth}
  \vspace*{-\intextsep}
  \begin{reactcode}
let Counter _ =
  print 0;
  let (s, setS) = useState 1 in
  [s, button (fun _ ->
   setS (fun s -> (print 1; s));
   setS (fun s -> (print 2; s+1));
   setS (fun s -> (print 3; s+1)))];;
Counter ()
  \end{reactcode}
\end{wrapfigure}
\noindent In our semantics, clicking the button prints~\verb+0\0\1\2\3+.
The first~\verb+0+ is printed during the initial render.
When the button is clicked after the render, the event handler queues the state updates without printing anything.
Then, during the second render triggered by the state updates, another~\verb+0+ is printed.
Finally, as the updates are applied,~\verb+1\2\3+ are printed.

In contrast, React prints~\verb+0\1\2\0\3+ due to the optimization.
The first~\verb+0+ is printed during the initial render.
Upon clicking the button, the first two updates are evaluated immediately:~\verb+1+ and~\verb+2+ are printed as React checks whether the state has changed.
The second update returns a new state, triggering a re-render.
The third update is skipped at this point because a re-render is already scheduled.
During the re-render, the second~\verb+0+ is printed, followed by~\verb+3+ from the queued third update.

We could have included this optimization for completeness, but we chose not to in order to maintain conciseness.
However, this omission does not compromise correctness as long as applications conform to the constraints specified in the React documentation.

Indeed, we show that the optimization preserves program behavior in \cref{thm:simtrans} (under certain conditions).

The optimization can be understood as partially applying queued state updates without re-executing the component body.
Assuming these updates are pure (\cref{def:purity}) as required by the React documentation, this corresponds to a form of partial normalization,
in which only a prefix of the pure updates is applied.
We formalize the relationship between the original and optimized tree memory using the notion of \emph{similarity} ($m \approx_t m'$; \cref{def:similarity}),
which is defined in terms of \emph{normalization} (\cref{def:normalization})---a process of applying pure updates in advance.

\begin{defn}[Purity]\label{def:purity}
  A closure~$\<\lfun x e, \sigma\>$ is \emph{pure} iff
  the application of any value~$v$ does not cause any side effects.
  That is, for all~$\Sigma$ and~$v$, we have~$\Step[\Sigma,\sigma[x\mapsto v]]{e}{\_,\Sigma, []}$
\end{defn}

\begin{defn}[Normalization]\label{def:normalization}
  A \emph{normalization} of a state store entry is defined as the result of applying its pure prefix of state updates, with its resulting decision collected.
  Let~$l$ be the length of such a prefix.
  Then
  \begin{multline*}
    \sem{normalize}\bigl(\bigl\{\fv{val}{v_0},\fv{sttq}{\bigl[\Overline{\<\lfun{x_i}{e_i},\sigma_i\>}\bigr]_{i=1}^n}\bigr\}\bigr)
    = \bigl\<\bigl\{ \fv{val}{v_l},
         \fv{sttq}{\bigl[\Overline{\<\lfun{x_i}{e_i},\sigma_i\>}\bigr]_{i=l+1}^n} \bigr\}, d\bigr\> \\
    \text{where }\bigl(
      \Step[\pi,\sigma_i[x_i \mapsto v_{i-1}]]{e_i}{v_i,\pi, []}
    \bigr)_{i=1}^{l}\\
    \text{and }d =
      \scondp{l \ne n \scondq \{\dec{Check}\}
      \scondcolon v_0 \nequiv v_l \scondq \{\dec{Check},\dec{Effect}\}
      \scondcolon \{\}}
  \end{multline*}

  We extend this notion to a state store~$\rho$ and a view~$\pi$, i.e., $\sem{normalize}(\rho)$ normalizing all of its entries, and~$\sem{normalize}(\pi)$ normalizing its state store:
  \begin{alignat*}{2}
    \sem{normalize}(\rho) &= \Bigl\<
        [\Overline{\ell \mapsto r_\ell}]_{\ell \in \domain \rho},
        \smashoperator{\bigcup_{\ell \in \domain \rho}} d_\ell
    \Bigr\> && \text{ where } \sem{normalize}(\rho[\ell]) = \<r_\ell,d_\ell\>, \\
    \sem{normalize}(\pi) &= 
      \pi \bigl\{ \fv{sttst}{\rho'},\fv{dec}{\pi.\field{dec} \setminus \{\dec{Check}\} \cup d} \bigr\} && \text{ where } \<\rho',d\> = \sem{normalize}(\pi.\field{sttst}).
    \qedhere 
  \end{alignat*}
\end{defn}

\begin{defn}[Similarity]\label{def:similarity}
  Two views are \emph{similar} iff they are equal under normalization:
  \[
    \pi \approx \pi'
    \triangleiff
    \sem{normalize}(\pi) = \sem{normalize}(\pi')
  \]

  We extend this notion to memories, i.e., $m$~and~$m'$ are \emph{$t$-similar} iff the descendants of~$t$ in $m$~and~$m'$ are all similar and the rest of them are equal:
  \[
    m \approx_t m'
    \triangleiff
    {}\land \begin{lgathered}
      \forall p \in \sem{reachable}(m, t),\ m[p] \approx m'[p] \\
      \forall p \notin \sem{reachable}(m, t),\ m[p] = m'[p] \\
    \end{lgathered}
  \]
  where $\sem{reachable}(m, t)$ is the set of all paths in~$m$ that is reachable from the tree~$t$.
\end{defn}

We now state that evaluating a view similar to the original yields the same result (\cref{lem:simevalbody}), provided the view is \emph{valid} (\cref{def:validity}).

\begin{defn}[Validity]\label{def:validity}
  A view is \emph{valid} if it has only the states that are present in its component body.
  More precisely,~$\pi$ is \emph{valid under~$\delta$} iff the domain of $\pi$'s state store is exactly the set of state labels of~$e$, where~$\pi.\field{spec} = \<C, v\>$ and~$\delta[C] = \<\lfun{x}{e}, \sigma\>$.
  We extend this notion to tree memory: A tree memory~$m$ is valid under~$\delta$ iff its views are all valid under~$\delta$.
\end{defn}

Note that all views encountered during execution are valid.
This follows from a syntactic restriction on component bodies: Hooks must appear only at the top level.
This ensures that each \reacttt{useState} call in a component body is executed exactly once during initialization,
so all views produced by~$\sem{init}$ are valid.

\begin{restatable}[Similar Evaluations of Component Body]{lem}{simevallem}\label{lem:simevalbody}
  Evaluating the component body of similar views produces the same value, view, and output buffer.
  That is, for $\pi$ valid under $\delta$ where $\pi.\field{spec} = \<C,v\>$ and $\delta[C] = \<\lambda x.e,\sigma\>$,
  if $\Step*<\phase{Succ}>[\pi,\sigma[x\mapsto v]]{e}{v',\pi', \omega}$,
  we have $\Step*<\phase{Succ}>[\hat \pi,\sigma[x\mapsto v]]{e}{v',\pi', \omega}$ for all~$\hat\pi \approx \pi$.
\end{restatable}

Building on \cref{lem:simevalbody}, we conclude that the optimization---which replaces some views in memory with similar ones---preserves program behavior.
Note that the output buffer~$\omega$ may differ if component bodies perform printing,
since the optimization may omit their evaluation entirely.
Nevertheless, the resulting memory remains identical, regardless of whether the optimization is applied.

\begin{restatable}[Similar Transitions]{thm}{simtransthm}\label{thm:simtrans}
  If a program transitions to a state in check mode~$\checkmode$ during execution, replacing it with a similar state results in the same final state as the original transition.
  That is, if $\<e,\delta\> \smallstep^* \<t,m,\omega,\delta,\checkmode\> \smallstep \<t,m',\omega',\delta,\mu\>$, then for any $\hat m$ such that $m \approx_t \hat m$, we have $\<t,\hat m,\omega,\delta,\checkmode\> \smallstep \<t,m',\omega'',\delta,\mu\>$.
  We also have~$\omega' = \omega''$ when the component bodies do not print.
\end{restatable}

Full proofs are provided in~\S B.

\subsection{Conformance Test Suite}\label{subsec:tests}
Our test suite, which covers all~44 evaluation rules\footnote{See~\S A.2 for the complete set of rules.} comprising the operational semantics~(\cref{subsec:sem}), contains 38~tests that cover 18~scenarios shown in \cref{tab:test} that compare the behaviors of programs under \Rtrace{} against React's behavior.
We have constructed the tests ourselves, testing the \Rtrace{} interpreter (\cref{sec:impl}) as well.

\begin{table}[t]
  \caption{Empirical validation of \Rtrace{} against React behavior.}\label{tab:test}
  \centering\smaller[2]
  \begin{tabular}{rlrc}
    \toprule
        & \textbf{Scenarios} & \textbf{Tests \#} & \textbf{\Rtrace{}} \\
    \midrule
     S1 & No re-render w/o a setter call & 6 & \checkmark \\
     S2 & Retries ($0<n<25$) w/ setter call during body eval & 4 & \checkmark \\
     S3 & Infinite retries ($n\ge25$) w/ setter call during body eval & 1 & \checkmark \\
     S4 & No re-render w/o Effects w/ setter call during body eval & 1 & \checkmark \\
     S5 & No re-render w/ Effect w/o setter call & 2 & \checkmark \\
     S6 & No re-render w/ Effect w/ id setter call & 1 & \checkmark \\
     S7 & No re-render w/ Effect w/ setter calls composing to id & 2 & \checkmark \\
     S8 & Re-renders ($0<n<100$) w/ Effect w/ setter call & 16 & \checkmark \\
     S9 & Infinite re-renders ($n\ge100$) w/ Effect w/ diverging setter call & 2 & \checkmark \\
    S10 & Re-render w/ child updating parent during Effect & 2 & \checkmark \\
    S11 & Re-render w/ sibling updating another during Effect & 1 & \checkmark \\
    S12 & Error w/ child updating parent during body eval & 1 & \checkmark \\
    S13 & Non-trivial reconciliation & 4 & \checkmark \\
    S14 & No re-render w/ direct object update & 1 & \checkmark \\
    S15 & Re-render w/ idle but parent updates & 2 & \checkmark \\
    S16 & User event sequence & 6 & \checkmark \\
    S17 & Re-render w/ setter call from user event & 4 & \checkmark\rlap{${}^\dagger$} \\
    S18 & Recursive view hierarchy & 2 & \checkmark \\
    \bottomrule
    && $38\mathrlap{{}^*}$ & \checkmark \\
    &&&\hfill\llap{\smaller[1]${}^*$Some tests cover multiple scenarios.}\\
    &&&\hfill\llap{\smaller[1]${}^\dagger$React's optimization changes some execution orders.}
  \end{tabular}
\end{table}

It is difficult to unit test a single evaluation rule independently because running a React program inevitably leaves footprints on multiple rules, and we have chosen the test scenarios to exhibit the pitfalls explained in~\cref{sec:pit}, as well as to check the core aspects of React including render counts, side effect ordering, reconciliation, and event handling.
For instance, although one of our test cases\footnote{\texttt{effect\_queue\_gets\_flushed\_on\_retry}} checks whether the effect queue gets flushed on retry, effectively testing rule \Rule{EvalMult}~(\cref{fig:evalmult}), it is simply categorized as S2~and~S8 in \cref{tab:test}.

Notably, we test the boundaries of React that even experienced developers may find confusing.
For instance, S12~in \cref{tab:test} covers attempting to update a parent's state during the evaluation of its child, which correctly produces an error in both implementations.
Other scenarios include unnecessary re-rendering that modifies the view hierarchy~(S8--S11,~\cref{tab:test}; similar to examples in~\cref{subsec:unnec}), constructing complex UI with recursive components~(S18,~\cref{tab:test}), and more.

The tests consist of equivalent components implemented twice---once in \Rtrace{} and once in React---then comparing that they exhibit the same behavior.
For example, to empirically check infinite retries~(S3,~\cref{tab:test}) on the React-side, we install an error boundary \citep{Met25Component} to catch an exception when a component reaches React's hard-coded limit of 25~retries.
On the \Rtrace{}-side, we set the retry threshold to~25 and see whether the component did not stop rendering.
For tests measuring the number of render cycles on the React-side, we indirectly measure the count by counting the prints inside an Effect.

Almost all tests show identical behavior between \Rtrace{} and React's runtime, with one minor difference due to the optimization performed by React as discussed in~\cref{subsubsec:opt}.
We might as well include this optimization in \Rtrace{} by adding a special flag and modifying \Rule{AppSetNormal}, but we deliberately chose not to in order to keep the clarity of our semantics.

\Rtrace{} does not model a specific version of React: We have tested against the latest versions of every major React release since Hooks were introduced in~\href{https://www.npmjs.com/package/react/v/16.8.0}{16.8\,(Feb~2019)}.  
All test cases have been reproduced in versions~\href{https://www.npmjs.com/package/react/v/16.14.0}{16.14.0\,(Oct~2020)}, \href{https://www.npmjs.com/package/react/v/17.0.2}{17.0.2\,(Mar~2021)}, \href{https://www.npmjs.com/package/react/v/18.3.1}{18.3.1\,(Apr~2024)}, and \href{https://www.npmjs.com/package/react/v/19.1.0}{19.1.0\,(Mar~2025)}.

The test suite is available alongside the \Rtrace{} interpreter at \href{https://github.com/Zeta611/react-trace}{\gh{Zeta611/react-trace}}.

\section{Discussion}\label{sec:discussion}
\paragraph{Extensions to Other Hooks}
While our semantics focuses on the two most prevalent Hooks \reacttt{useState} and \reacttt{useEffect}, it's designed to be extensible.
Incorporating additional Hooks builds upon existing machinery for bookkeeping and lifecycle modes.

Hooks are mostly state managing (\reacttt{useState}, \reacttt{useRef}, \reacttt{useMemo}, \reacttt{useContext}) or side effect performing (\reacttt{useEffect}, \reacttt{useLayoutEffect}, \reacttt{useInsertionEffect}), which can be supported by extending view record~$\pi$ and/or inserting modes in render transitions.
Users can also build custom Hooks by combining existing Hooks in a function, which should be \textbf{\texttt{use}}-prefixed.

We sketch how to extend the semantics to other Hooks:
\begin{description}
  \item[$\textbf{\texttt{useRef}}^\ell$] Add $\pi.\field{refst}$ storing refs~$[\Overline{\ell \mapsto v}]$.
    Unlike \reacttt{useState}, mutating a ref does not trigger re-renders, so no \dec{Check} decision is added when mutated.
    In \phase{Init} phase, the ref is initialized; in \phase{Succ} phase, the same ref is returned.
  \item[$\textbf{\texttt{useMemo}}^\ell$] Add $\pi.\field{memost}$ storing values and their dependencies~$[\Overline[\ell]{\ell \mapsto \{\field{val}: v, \field{deps}: [\Overline{v}]\}}]$.
    Similar to \reacttt{useState}, in \phase{Init} the value is computed from the provided function; in \phase{Succ} the value is recomputed only when dependencies differ from the stored ones.
  \item[\textbf{\texttt{useContext}}] Add $\pi.\field{ctxst}$ storing a list of consumed contexts and tree memory~$m$ is traversed upward to find the nearest context provider.
    When a provider's value changes, all consuming views are marked with \dec{Check} to trigger re-renders.
  \item[\textbf{\texttt{useLayoutEffect}}] Add mode~$\rendermode'$ into render transitions (\cref{fig:step}).
    \Rule{StepInit} enters~$\rendermode'$ after the initial render; a new rule \Rule{StepLayoutEffect} processes LayoutEffects determining the transition~$\<t, m, \omega, \delta, \rendermode'\> \smallstep \<t, m', \omega', \delta, \checkmode\>$.
  \item[\textbf{\texttt{useInsertionEffect}}] Similar to \reacttt{useLayoutEffect} but runs before.
    State updates are forbidden in InsertionEffects by the React runtime, hence the pitfalls described in \cref{sec:pit} cannot occur. 
  \item[Custom Hooks] Supported by allowing user function definitions~(\cref{fig:syntax}) and \dom{DefTable}~(\cref{subsec:dom}).
    No additional machinery is required as custom Hooks are compositions of primitive Hooks.
\end{description}

\paragraph{Verifying the React Compiler}
Our formalization of Hooks is particularly timely as the React team is currently developing the React compiler \citep{Met25Compiler}, where various optimizations can be verified using our semantics as in~\cref{subsubsec:opt}.
The compiler optimizes React programs by eliminating unnecessary re-renders through memoization.
However, without a proper formal semantics of Hooks, there is no rigorous way to verify that the compiler preserves program behavior.
Our semantics can thus serve as a foundation for reasoning about their correctness.

\paragraph{Beyond React}
Our work can be extended to accommodate other reactive UI frameworks as well.
There is a plethora of GUI frameworks---React, Preact (\href{https://preactjs.com/}{preactjs.com}), Dioxus (\href{https://dioxuslabs.com/}{dioxuslabs.com}), Solid (\href{https://www.solidjs.com/}{solidjs.com}), Svelte (\href{https://svelte.dev/}{svelte.dev}), and Leptos (\href{https://leptos.dev/}{leptos.dev})---with variations in their reactivity models.

These frameworks can be categorized based on three properties: (a)~whether they re-read component specifications for re-rendering, (b)~how state updates are processed (queued or immediate), and (c)~which reactivity primitives they employ.

Reactive primitives include Hooks (where frameworks schedule update checks), signals (where state updates propagate actively), and compiler-assisted methods (where dependencies are statically resolved).
By parameterizing our semantics with appropriate timing and semantic operators, we capture these variations.
A full comparison of the frameworks is provided in~\S C. 

As a reference to how similar these frameworks can be, compare the Dioxus code on the left (where Rules of Hooks apply as well! \citep{Dio25Hooks}) and the \Rtrace{} code on the right.
\begin{center}
  \begin{minipage}{0.45\linewidth}
    \begin{reactcode}
pub fn Counter() -> Element {
  let mut count = use_signal(|| 0);
  rsx! { button { onclick: move |_|
    count += 1, "{count}" } }
    \end{reactcode}
  \end{minipage}%
  \begin{minipage}{0.45\linewidth}
    \begin{reactcode}
let Counter _ =
  let (count, setCount) = useState 0 in
  [button (fun _ -> 
    setCount(fun count -> count+1)), count]
    \end{reactcode}
  \end{minipage}
\end{center}

\section{Related Work}\label{sec:related}
\paragraph{Research on React}
Our work presents the first formal semantics for React's function components and Hooks, which have become the standard for modern React applications.
Moreover, we deliberately decouple the semantics of React and Hooks from JS, enabling reasoning about Hooks independently from the host language.
This approach contrasts with previous work on classed-based React components by \citet{MadLhoTip20Semantics}, who formalized a core calculus~\Lreact{} upon~\Ljs{} \citep{GuhSafKri10Essence} that captures React's component lifecycle and reconciliation.

While our research aims to faithfully model the actual React runtime, \citet{CriKri24Core} took a more general approach with their document calculus.
Their semantics for general document languages demonstrates how features like document references interact with documents written in a React-like language using a simplified runtime, rather than attempting to capture the full complexity of React's behavior.

\paragraph{React Development Tools}
\Rtrace{} provides a theoretical foundation for building practical tools grounded on the semantics of Hooks, unlike existing approaches limited to runtime detection or syntactic checking.
This semantics-based approach contrasts with tools like ESLint (\href{https://eslint.org}{eslint.org}), a \emph{de facto} standard static analyzer that syntactically checks code to prevent basic errors when writing React components.
Similarly, runtime libraries such as \texttt{why-did-you-render} (\href{https://github.com/welldone-software/why-did-you-render}{\gh{welldone-software/why-did-you-render}}) and React Scan (\href{https://react-scan.com/}{react-scan.com}) detect (re-)renders by monkey patching and inspecting the React library at runtime but lack the formal framework to explain why these behaviors occur.

\paragraph{Functional Reactive GUI Frameworks}
React borrows concepts from functional reactive programming (FRP) \citep{EllHud97Functional} but implements them quite differently.
While React embraces declarative components, it requires explicit state management via \reacttt{useState} and provides escape hatches through \reacttt{useEffect}.
This differs significantly from FRP web frameworks which use signals as reactive primitives.
Notable FRP web frameworks include Flapjax \citep{MeyGuhBasCooGreBroKri09Flapjax}, Ur/Web \citep{Chl15UrWeb,Chl15Optimizing}, and the early versions of Elm \citep{CzaCho13Asynchronous}. 

\paragraph{The Elm Architecture}
The Elm programming language (\href{https://elm-lang.org/}{elm-lang.org}) has moved on from FRP and settled on \emph{The Elm Architecture (TEA)} or the \emph{Model-View-Update (MVU)} pattern, which has been formalized as~\Lmvu{} \citep{Fow20ModelViewUpdateCommunicate}.
In TEA, there is a clear separation between a model and a view, where a message is dispatched from the view to update the model through an update function, creating a unidirectional data flow \citep{Fel20Elm}.
One can follow a similar pattern in React if a state management library like Redux (\href{https://redux.js.org/}{redux.js.org}) is used.

\paragraph{Multi-tier Programming}
In practice, React applications span multiple \emph{tiers}---client, server, and database---requiring developers to manually maintain coherency between them.
Full stack React frameworks such as Next.js (\href{https://nextjs.org/}{nextjs.org}) and Remix (\href{https://remix.run/}{remix.run}) borrow ideas from \emph{tierless} or \emph{multi-tier}~(MT) languages to allow developers to develop both the client- and server-side in a single program.
MT languages like Ur/Web \citep{Chl15UrWeb,Chl15Optimizing} (an ML-like language with whole-program optimization), \textsc{Eliom} \citep{RadVouBal16Eliom,RadPapVouBal16Eliom} (an extension of OCaml focusing on modularity), Hop \citep{SerGalLoi06Hop,BouLuoRezSer12Reasoning} (a dynamically typed Scheme-based framework), and Links \citep{CooLinWadYal07Links} (an ML-like language with an effect system that compiles to SQL and JS) enable higher-level reasoning about the whole system while providing stronger safety guarantees \citep{WeiWirSal21Survey}.

\paragraph{JavaScript and Web Semantics}
Our work extends the rich tradition of formalizing web technologies by providing a formal model of React Hooks, complementing existing work on web semantics.
Numerous formal models of JS have been proposed \citep{GuhSafKri10Essence,MafMitTal08Operational,ParSteAndRos15KJS,PolCarLerPomKri12Tested,FraMakNauWooGar18JaVerT,MadLhoTip17Model}, including the recent work on mechanized extraction of the prose ECMA-262 specification by \citet{ParParAnRyu21JISET}.
WebSpec \citep{VerFarBerTemSquMaf23WebSpec} provides a formal browser model in Rocq (formerly known as Coq) to verify browser security.
Several works \citep{PanTor16Automated,PanErnTatKam19Modular,PanGelErnTatKam18Verifying} have formalized the browser layout algorithm and the CSS semantics using SMT-based encoding to verify web page layouts.

\section{Conclusion and Future Work}\label{sec:conclusion}
We have presented \Rtrace{}, an operational semantics for helping React developers understand the behavior of Hooks.
Render-related UI bugs often puzzle developers due to the peculiar semantics of Hooks.
We have captured the essence of Hooks in \Rtrace{}, which clearly explains subtle pitfalls that occur when multiple Hooks interact with each other.
Moreover, developers can use our \Rtrace-interpreter and visualizer to inspect how their programs render step-by-step, making the opaque render process transparent.

As a foundation for future work, \Rtrace{} opens several research directions.
First, our semantics enables semantic-based static analyzers and type systems that go beyond the syntactic or runtime checks of current tools, catching subtle bugs at compile time.
Second, as React evolves with features like React compiler and server components, \Rtrace{} provides a formal basis for verifying their correctness.
Third, \Rtrace{} can be refined for pedagogical use in building a correct mental model of Hooks.
Finally, integrating \Rtrace{} with existing web semantics can enable end-to-end verification of web programs, from component behavior to final rendering.

\section*{Data-Availability Statement}
The source code and test suite for the \Rtrace{} interpreter and visualizer are available on GitHub at \href{https://github.com/Zeta611/react-trace}{\gh{Zeta611/react-trace}}, archived on Zenodo \citep{LeeAhn25Zenodo}, and hosted online at \href{https://react-trace.vercel.app/}{\nolinkurl{react-trace.vercel.app}}.

\begin{acks}
  We are grateful to Joonhyup Lee, Gyuhyeok Oh, Lauren Minjung Kwon, and Hyeongseo Yoo for valuable comments,
  Will Crichton for early feedback,
  and Sukyoung Ryu for helpful suggestions.
  We thank Susana Ang\'elica Balderas Chiw for the logo in \cref{fig:vis}.
  This work was supported by BK21 FOUR Intelligence Computing (Dept. of Computer Science and Engineering, SNU) funded by National Research Foundation of Korea (NRF) (4199990214639), Greenlabs Co., Ltd. (0536-20220078), and Samsung Electronics Co., Ltd. (0536-20230088).
\end{acks}

\bibliographystyle{ACM-Reference-Format}
\bibliography{references}


\begin{thebibliography}{51}


\ifx \showCODEN    \undefined \def \showCODEN     #1{\unskip}     \fi
\ifx \showISBNx    \undefined \def \showISBNx     #1{\unskip}     \fi
\ifx \showISBNxiii \undefined \def \showISBNxiii  #1{\unskip}     \fi
\ifx \showISSN     \undefined \def \showISSN      #1{\unskip}     \fi
\ifx \showLCCN     \undefined \def \showLCCN      #1{\unskip}     \fi
\ifx \shownote     \undefined \def \shownote      #1{#1}          \fi
\ifx \showarticletitle \undefined \def \showarticletitle #1{#1}   \fi
\ifx \showURL      \undefined \def \showURL       {\relax}        \fi
\providecommand\bibfield[2]{#2}
\providecommand\bibinfo[2]{#2}
\providecommand\natexlab[1]{#1}
\providecommand\showeprint[2][]{arXiv:#2}

\bibitem[Boudol et~al\mbox{.}(2012)]%
        {BouLuoRezSer12Reasoning}
\bibfield{author}{\bibinfo{person}{G\'{e}rard Boudol},
  \bibinfo{person}{Zhengqin Luo}, \bibinfo{person}{Tamara Rezk}, {and}
  \bibinfo{person}{Manuel Serrano}.} \bibinfo{year}{2012}\natexlab{}.
\newblock \showarticletitle{Reasoning about Web Applications: An Operational
  Semantics for HOP}.
\newblock \bibinfo{journal}{\emph{ACM Trans. Program. Lang. Syst.}}
  \bibinfo{volume}{34}, \bibinfo{number}{2}, Article \bibinfo{articleno}{10}
  (\bibinfo{date}{June} \bibinfo{year}{2012}), \bibinfo{numpages}{40}~pages.
\newblock
\showISSN{0164-0925}
\href{https://doi.org/10.1145/2220365.2220369}{doi:\nolinkurl{10.1145/2220365.2220369}}


\bibitem[Chlipala(2015a)]%
        {Chl15Optimizing}
\bibfield{author}{\bibinfo{person}{Adam Chlipala}.}
  \bibinfo{year}{2015}\natexlab{a}.
\newblock \showarticletitle{An Optimizing Compiler for a Purely Functional
  Web-Application Language}. In \bibinfo{booktitle}{\emph{Proceedings of the
  20th ACM SIGPLAN International Conference on Functional Programming}}
  (Vancouver, BC, Canada) \emph{(\bibinfo{series}{ICFP '15})}.
  \bibinfo{publisher}{Association for Computing Machinery},
  \bibinfo{address}{New York, NY, USA}, \bibinfo{pages}{10–21}.
\newblock
\showISBNx{9781450336697}
\href{https://doi.org/10.1145/2784731.2784741}{doi:\nolinkurl{10.1145/2784731.2784741}}


\bibitem[Chlipala(2015b)]%
        {Chl15UrWeb}
\bibfield{author}{\bibinfo{person}{Adam Chlipala}.}
  \bibinfo{year}{2015}\natexlab{b}.
\newblock \showarticletitle{Ur/Web: A Simple Model for Programming the Web}. In
  \bibinfo{booktitle}{\emph{Proceedings of the 42nd Annual ACM SIGPLAN-SIGACT
  Symposium on Principles of Programming Languages}} (Mumbai, India)
  \emph{(\bibinfo{series}{POPL '15})}. \bibinfo{publisher}{Association for
  Computing Machinery}, \bibinfo{address}{New York, NY, USA},
  \bibinfo{pages}{153–165}.
\newblock
\showISBNx{9781450333009}
\href{https://doi.org/10.1145/2676726.2677004}{doi:\nolinkurl{10.1145/2676726.2677004}}


\bibitem[Cooper et~al\mbox{.}(2007)]%
        {CooLinWadYal07Links}
\bibfield{author}{\bibinfo{person}{Ezra Cooper}, \bibinfo{person}{Sam Lindley},
  \bibinfo{person}{Philip Wadler}, {and} \bibinfo{person}{Jeremy Yallop}.}
  \bibinfo{year}{2007}\natexlab{}.
\newblock \showarticletitle{Links: Web Programming Without Tiers}. In
  \bibinfo{booktitle}{\emph{Formal Methods for Components and Objects (FMCO
  '06)}} \emph{(\bibinfo{series}{Lecture Notes in Computer Science},
  Vol.~\bibinfo{volume}{4709})}, \bibfield{editor}{\bibinfo{person}{Frank~S.
  de~Boer}, \bibinfo{person}{Marcello~M. Bonsangue}, \bibinfo{person}{Susanne
  Graf}, {and} \bibinfo{person}{Willem-Paul de~Roever}} (Eds.).
  \bibinfo{publisher}{Springer Berlin Heidelberg}, \bibinfo{address}{Berlin,
  Heidelberg}, \bibinfo{pages}{266--296}.
\newblock
\href{https://doi.org/10.1007/978-3-540-74792-5_12}{doi:\nolinkurl{10.1007/978-3-540-74792-5_12}}


\bibitem[Cousot and Cousot(1977)]%
        {CouCou77Abstract}
\bibfield{author}{\bibinfo{person}{Patrick Cousot} {and}
  \bibinfo{person}{Radhia Cousot}.} \bibinfo{year}{1977}\natexlab{}.
\newblock \showarticletitle{Abstract Interpretation: A Unified Lattice Model
  for Static Analysis of Programs by Construction or Approximation of
  Fixpoints}. In \bibinfo{booktitle}{\emph{Proceedings of the 4th ACM
  SIGACT-SIGPLAN Symposium on Principles of Programming Languages}} (Los
  Angeles, California) \emph{(\bibinfo{series}{POPL '77})}.
  \bibinfo{publisher}{Association for Computing Machinery},
  \bibinfo{address}{New York, NY, USA}, \bibinfo{pages}{238–252}.
\newblock
\showISBNx{9781450373500}
\href{https://doi.org/10.1145/512950.512973}{doi:\nolinkurl{10.1145/512950.512973}}


\bibitem[Crichton and Krishnamurthi(2024)]%
        {CriKri24Core}
\bibfield{author}{\bibinfo{person}{Will Crichton} {and}
  \bibinfo{person}{Shriram Krishnamurthi}.} \bibinfo{year}{2024}\natexlab{}.
\newblock \showarticletitle{A Core Calculus for Documents: Or, Lambda: The
  Ultimate Document}.
\newblock \bibinfo{journal}{\emph{Proc. ACM Program. Lang.}}
  \bibinfo{volume}{8}, \bibinfo{number}{POPL}, Article \bibinfo{articleno}{23}
  (\bibinfo{date}{Jan.} \bibinfo{year}{2024}), \bibinfo{numpages}{28}~pages.
\newblock
\href{https://doi.org/10.1145/3632865}{doi:\nolinkurl{10.1145/3632865}}


\bibitem[Czaplicki and Chong(2013)]%
        {CzaCho13Asynchronous}
\bibfield{author}{\bibinfo{person}{Evan Czaplicki} {and}
  \bibinfo{person}{Stephen Chong}.} \bibinfo{year}{2013}\natexlab{}.
\newblock \showarticletitle{Asynchronous Functional Reactive Programming for
  GUIs}. In \bibinfo{booktitle}{\emph{Proceedings of the 34th ACM SIGPLAN
  Conference on Programming Language Design and Implementation}} (Seattle,
  Washington, USA) \emph{(\bibinfo{series}{PLDI~'13})}.
  \bibinfo{publisher}{Association for Computing Machinery},
  \bibinfo{address}{New York, NY, USA}, \bibinfo{pages}{411–422}.
\newblock
\showISBNx{9781450320146}
\href{https://doi.org/10.1145/2491956.2462161}{doi:\nolinkurl{10.1145/2491956.2462161}}


\bibitem[{Dioxus Labs}(2024)]%
        {Dio25Hooks}
\bibfield{author}{\bibinfo{person}{{Dioxus Labs}}.}
  \bibinfo{year}{2024}\natexlab{}.
\newblock \bibinfo{title}{Hooks and component state}.
\newblock
  \bibinfo{howpublished}{\url{https://dioxuslabs.com/learn/0.6/reference/hooks/}}.
\newblock
\newblock
\shownote{Accessed: 2025-03-25}.


\bibitem[Elliott and Hudak(1997)]%
        {EllHud97Functional}
\bibfield{author}{\bibinfo{person}{Conal Elliott} {and} \bibinfo{person}{Paul
  Hudak}.} \bibinfo{year}{1997}\natexlab{}.
\newblock \showarticletitle{Functional Reactive Animation}. In
  \bibinfo{booktitle}{\emph{Proceedings of the Second ACM SIGPLAN International
  Conference on Functional Programming}} (Amsterdam, The Netherlands)
  \emph{(\bibinfo{series}{ICFP '97})}. \bibinfo{publisher}{Association for
  Computing Machinery}, \bibinfo{address}{New York, NY, USA},
  \bibinfo{pages}{263–273}.
\newblock
\showISBNx{0897919181}
\href{https://doi.org/10.1145/258948.258973}{doi:\nolinkurl{10.1145/258948.258973}}


\bibitem[Feldman(2020)]%
        {Fel20Elm}
\bibfield{author}{\bibinfo{person}{Richard Feldman}.}
  \bibinfo{year}{2020}\natexlab{}.
\newblock \bibinfo{booktitle}{\emph{Elm in Action}}.
\newblock \bibinfo{publisher}{Manning Publications}, \bibinfo{address}{Shelter
  Island, NY}. 344 pages.
\newblock
\showISBNx{9781617294044}
\urldef\tempurl%
\url{https://www.manning.com/books/elm-in-action}
\showURL{%
\tempurl}


\bibitem[Fowler(2020)]%
        {Fow20ModelViewUpdateCommunicate}
\bibfield{author}{\bibinfo{person}{Simon Fowler}.}
  \bibinfo{year}{2020}\natexlab{}.
\newblock \showarticletitle{Model-View-Update-Communicate: Session Types Meet
  the {Elm} Architecture}. In \bibinfo{booktitle}{\emph{34th European
  Conference on Object-Oriented Programming (ECOOP 2020)}}
  \emph{(\bibinfo{series}{Leibniz International Proceedings in Informatics
  (LIPIcs)}, Vol.~\bibinfo{volume}{166})},
  \bibfield{editor}{\bibinfo{person}{Robert Hirschfeld} {and}
  \bibinfo{person}{Tobias Pape}} (Eds.). \bibinfo{publisher}{Schloss Dagstuhl
  -- Leibniz-Zentrum f{\"u}r Informatik}, \bibinfo{address}{Dagstuhl, Germany},
  \bibinfo{pages}{14:1--14:28}.
\newblock
\showISBNx{978-3-95977-154-2}
\showISSN{1868-8969}
\href{https://doi.org/10.4230/LIPIcs.ECOOP.2020.14}{doi:\nolinkurl{10.4230/LIPIcs.ECOOP.2020.14}}


\bibitem[Fragoso~Santos et~al\mbox{.}(2017)]%
        {FraMakNauWooGar18JaVerT}
\bibfield{author}{\bibinfo{person}{Jos\'{e} Fragoso~Santos},
  \bibinfo{person}{Petar Maksimovi\'{c}}, \bibinfo{person}{Daiva
  Naud\v{z}i\={u}nien\.{e}}, \bibinfo{person}{Thomas Wood}, {and}
  \bibinfo{person}{Philippa Gardner}.} \bibinfo{year}{2017}\natexlab{}.
\newblock \showarticletitle{JaVerT: JavaScript Verification Toolchain}.
\newblock \bibinfo{journal}{\emph{Proc. ACM Program. Lang.}}
  \bibinfo{volume}{2}, \bibinfo{number}{POPL}, Article \bibinfo{articleno}{50}
  (\bibinfo{date}{Dec.} \bibinfo{year}{2017}), \bibinfo{numpages}{33}~pages.
\newblock
\href{https://doi.org/10.1145/3158138}{doi:\nolinkurl{10.1145/3158138}}


\bibitem[Guha et~al\mbox{.}(2010)]%
        {GuhSafKri10Essence}
\bibfield{author}{\bibinfo{person}{Arjun Guha}, \bibinfo{person}{Claudiu
  Saftoiu}, {and} \bibinfo{person}{Shriram Krishnamurthi}.}
  \bibinfo{year}{2010}\natexlab{}.
\newblock \showarticletitle{The Essence of JavaScript}. In
  \bibinfo{booktitle}{\emph{ECOOP 2010 -- Object-Oriented Programming}}
  \emph{(\bibinfo{series}{Lecture Notes in Computer Science},
  Vol.~\bibinfo{volume}{6183})}, \bibfield{editor}{\bibinfo{person}{Theo
  D'Hondt}} (Ed.). \bibinfo{publisher}{Springer Berlin Heidelberg},
  \bibinfo{address}{Berlin, Heidelberg}, \bibinfo{pages}{126--150}.
\newblock
\showISBNx{978-3-642-14107-2}
\href{https://doi.org/10.1007/978-3-642-14107-2_7}{doi:\nolinkurl{10.1007/978-3-642-14107-2_7}}


\bibitem[Kahn(1987)]%
        {Kah87Natural}
\bibfield{author}{\bibinfo{person}{Gilles Kahn}.}
  \bibinfo{year}{1987}\natexlab{}.
\newblock \showarticletitle{Natural Semantics}. In
  \bibinfo{booktitle}{\emph{STACS 87}} \emph{(\bibinfo{series}{Lecture Notes in
  Computer Science}, Vol.~\bibinfo{volume}{247})},
  \bibfield{editor}{\bibinfo{person}{Franz~J. Brandenburg},
  \bibinfo{person}{Guy Vidal-Naquet}, {and} \bibinfo{person}{Martin Wirsing}}
  (Eds.). \bibinfo{publisher}{Springer Berlin Heidelberg},
  \bibinfo{address}{Berlin, Heidelberg}, \bibinfo{pages}{22--39}.
\newblock
\showISBNx{978-3-540-47419-7}
\href{https://doi.org/10.1007/BFb0039592}{doi:\nolinkurl{10.1007/BFb0039592}}


\bibitem[Lee(2025)]%
        {Jay25Archive}
\bibfield{author}{\bibinfo{person}{Jay Lee}.} \bibinfo{year}{2025}\natexlab{}.
\newblock \bibinfo{title}{A screenshot of the StackOverflow search results of
  `"useEffect" "infinite"'}.
\newblock
  \bibinfo{howpublished}{\url{https://archive.org/details/useeffect-infinite}}.
\newblock
\newblock
\shownote{Accessed: 2025-03-24}.


\bibitem[Lee and Ahn(2025)]%
        {LeeAhn25Zenodo}
\bibfield{author}{\bibinfo{person}{Jay Lee} {and} \bibinfo{person}{Joongwon
  Ahn}.} \bibinfo{year}{2025}\natexlab{}.
\newblock \bibinfo{booktitle}{\emph{Artifact for
  ``{\textsc{\textsf{React-tRace}}}: A Semantics for Understanding React
  Hooks''}}.
\newblock
\href{https://doi.org/10.5281/zenodo.16916356}{doi:\nolinkurl{10.5281/zenodo.16916356}}


\bibitem[Madsen et~al\mbox{.}(2017)]%
        {MadLhoTip17Model}
\bibfield{author}{\bibinfo{person}{Magnus Madsen}, \bibinfo{person}{Ond\v{r}ej
  Lhot\'{a}k}, {and} \bibinfo{person}{Frank Tip}.}
  \bibinfo{year}{2017}\natexlab{}.
\newblock \showarticletitle{A Model for Reasoning About JavaScript Promises}.
\newblock \bibinfo{journal}{\emph{Proc. ACM Program. Lang.}}
  \bibinfo{volume}{1}, \bibinfo{number}{OOPSLA}, Article
  \bibinfo{articleno}{86} (\bibinfo{date}{Oct.} \bibinfo{year}{2017}),
  \bibinfo{numpages}{24}~pages.
\newblock
\href{https://doi.org/10.1145/3133910}{doi:\nolinkurl{10.1145/3133910}}


\bibitem[Madsen et~al\mbox{.}(2020)]%
        {MadLhoTip20Semantics}
\bibfield{author}{\bibinfo{person}{Magnus Madsen}, \bibinfo{person}{Ond\v{r}ej
  Lhot\'{a}k}, {and} \bibinfo{person}{Frank Tip}.}
  \bibinfo{year}{2020}\natexlab{}.
\newblock \showarticletitle{{A Semantics for the Essence of React}}. In
  \bibinfo{booktitle}{\emph{34th European Conference on Object-Oriented
  Programming (ECOOP 2020)}} \emph{(\bibinfo{series}{Leibniz International
  Proceedings in Informatics (LIPIcs)}, Vol.~\bibinfo{volume}{166})},
  \bibfield{editor}{\bibinfo{person}{Robert Hirschfeld} {and}
  \bibinfo{person}{Tobias Pape}} (Eds.). \bibinfo{publisher}{Schloss Dagstuhl
  -- Leibniz-Zentrum f{\"u}r Informatik}, \bibinfo{address}{Dagstuhl, Germany},
  \bibinfo{pages}{12:1--12:26}.
\newblock
\showISBNx{978-3-95977-154-2}
\showISSN{1868-8969}
\href{https://doi.org/10.4230/LIPIcs.ECOOP.2020.12}{doi:\nolinkurl{10.4230/LIPIcs.ECOOP.2020.12}}


\bibitem[Maffeis et~al\mbox{.}(2008)]%
        {MafMitTal08Operational}
\bibfield{author}{\bibinfo{person}{Sergio Maffeis}, \bibinfo{person}{John~C.
  Mitchell}, {and} \bibinfo{person}{Ankur Taly}.}
  \bibinfo{year}{2008}\natexlab{}.
\newblock \showarticletitle{An Operational Semantics for {JavaScript}}. In
  \bibinfo{booktitle}{\emph{Programming Languages and Systems (APLAS '08)}}
  \emph{(\bibinfo{series}{Lecture Notes in Computer Science},
  Vol.~\bibinfo{volume}{5356})}, \bibfield{editor}{\bibinfo{person}{Ganesan
  Ramalingam}} (Ed.). \bibinfo{publisher}{Springer Berlin Heidelberg},
  \bibinfo{address}{Berlin, Heidelberg}, \bibinfo{pages}{307--325}.
\newblock
\href{https://doi.org/10.1007/978-3-540-89330-1_22}{doi:\nolinkurl{10.1007/978-3-540-89330-1_22}}


\bibitem[{Meta Platforms, Inc.}(2022)]%
        {Met22JSX}
\bibfield{author}{\bibinfo{person}{{Meta Platforms, Inc.}}}
  \bibinfo{year}{2022}\natexlab{}.
\newblock \bibinfo{title}{JSX}.
\newblock \bibinfo{howpublished}{\url{https://facebook.github.io/jsx/}}.
\newblock
\newblock
\shownote{Accessed: 2025-03-24}.


\bibitem[{Meta Platforms, Inc.}(2025a)]%
        {Met25Component}
\bibfield{author}{\bibinfo{person}{{Meta Platforms, Inc.}}}
  \bibinfo{year}{2025}\natexlab{a}.
\newblock \bibinfo{title}{Component -- React}.
\newblock
  \bibinfo{howpublished}{\url{https://react.dev/reference/react/Component}}.
\newblock
\newblock
\shownote{Accessed: 2025-03-12}.


\bibitem[{Meta Platforms, Inc.}(2025b)]%
        {Met25Higher}
\bibfield{author}{\bibinfo{person}{{Meta Platforms, Inc.}}}
  \bibinfo{year}{2025}\natexlab{b}.
\newblock \bibinfo{title}{Higher-Order Components -- React}.
\newblock
  \bibinfo{howpublished}{\url{https://legacy.reactjs.org/docs/higher-order-components.html}}.
\newblock
\newblock
\shownote{Accessed: 2025-03-18}.


\bibitem[{Meta Platforms, Inc.}(2025c)]%
        {Met25Queueing}
\bibfield{author}{\bibinfo{person}{{Meta Platforms, Inc.}}}
  \bibinfo{year}{2025}\natexlab{c}.
\newblock \bibinfo{title}{Queueing a Series of State Updates -- React}.
\newblock
  \bibinfo{howpublished}{\url{https://react.dev/learn/queueing-a-series-of-state-updates}}.
\newblock
\newblock
\shownote{Accessed: 2025-03-12}.


\bibitem[{Meta Platforms, Inc.}(2025d)]%
        {Met25Compiler}
\bibfield{author}{\bibinfo{person}{{Meta Platforms, Inc.}}}
  \bibinfo{year}{2025}\natexlab{d}.
\newblock \bibinfo{title}{React Compiler -- React}.
\newblock \bibinfo{howpublished}{\url{https://react.dev/learn/react-compiler}}.
\newblock
\newblock
\shownote{Accessed: 2025-03-25}.


\bibitem[{Meta Platforms, Inc.}(2025e)]%
        {Met25Reconciliation}
\bibfield{author}{\bibinfo{person}{{Meta Platforms, Inc.}}}
  \bibinfo{year}{2025}\natexlab{e}.
\newblock \bibinfo{title}{Reconciliation – React}.
\newblock
  \bibinfo{howpublished}{\url{https://legacy.reactjs.org/docs/reconciliation.html}}.
\newblock
\newblock
\shownote{Accessed: 2025-07-23}.


\bibitem[{Meta Platforms, Inc.}(2025f)]%
        {Met25Render}
\bibfield{author}{\bibinfo{person}{{Meta Platforms, Inc.}}}
  \bibinfo{year}{2025}\natexlab{f}.
\newblock \bibinfo{title}{Render and Commit -- React}.
\newblock
  \bibinfo{howpublished}{\url{https://react.dev/learn/render-and-commit}}.
\newblock
\newblock
\shownote{Accessed: 2025-03-25}.


\bibitem[{Meta Platforms, Inc.}(2025g)]%
        {Met25Hooks}
\bibfield{author}{\bibinfo{person}{{Meta Platforms, Inc.}}}
  \bibinfo{year}{2025}\natexlab{g}.
\newblock \bibinfo{title}{Rules of Hooks -- React}.
\newblock
  \bibinfo{howpublished}{\url{https://react.dev/reference/rules/rules-of-hooks}}.
\newblock
\newblock
\shownote{Accessed: 2025-03-12}.


\bibitem[{Meta Platforms, Inc.}(2025h)]%
        {Met25Rules}
\bibfield{author}{\bibinfo{person}{{Meta Platforms, Inc.}}}
  \bibinfo{year}{2025}\natexlab{h}.
\newblock \bibinfo{title}{Rules of React -- React}.
\newblock \bibinfo{howpublished}{\url{https://react.dev/reference/rules}}.
\newblock
\newblock
\shownote{Accessed: 2025-03-12}.


\bibitem[{Meta Platforms, Inc.}(2025i)]%
        {Met25Syn}
\bibfield{author}{\bibinfo{person}{{Meta Platforms, Inc.}}}
  \bibinfo{year}{2025}\natexlab{i}.
\newblock \bibinfo{title}{Synchronizing with Effects -- React}.
\newblock
  \bibinfo{howpublished}{\url{https://react.dev/learn/synchronizing-with-effects}}.
\newblock
\newblock
\shownote{Accessed: 2025-03-12}.


\bibitem[{Meta Platforms, Inc.}(2025j)]%
        {Met25UseState}
\bibfield{author}{\bibinfo{person}{{Meta Platforms, Inc.}}}
  \bibinfo{year}{2025}\natexlab{j}.
\newblock \bibinfo{title}{useState -- React}.
\newblock
  \bibinfo{howpublished}{\url{https://react.dev/reference/react/useState}}.
\newblock
\newblock
\shownote{Accessed: 2025-03-25}.


\bibitem[Meyerovich et~al\mbox{.}(2009)]%
        {MeyGuhBasCooGreBroKri09Flapjax}
\bibfield{author}{\bibinfo{person}{Leo~A. Meyerovich}, \bibinfo{person}{Arjun
  Guha}, \bibinfo{person}{Jacob Baskin}, \bibinfo{person}{Gregory~H. Cooper},
  \bibinfo{person}{Michael Greenberg}, \bibinfo{person}{Aleks Bromfield}, {and}
  \bibinfo{person}{Shriram Krishnamurthi}.} \bibinfo{year}{2009}\natexlab{}.
\newblock \showarticletitle{Flapjax: A Programming Language for Ajax
  Applications}. In \bibinfo{booktitle}{\emph{Proceedings of the 24th ACM
  SIGPLAN Conference on Object Oriented Programming Systems Languages and
  Applications}} (Orlando, Florida, USA) \emph{(\bibinfo{series}{OOPSLA '09})}.
  \bibinfo{publisher}{Association for Computing Machinery},
  \bibinfo{address}{New York, NY, USA}, \bibinfo{pages}{1–20}.
\newblock
\showISBNx{9781605587660}
\href{https://doi.org/10.1145/1640089.1640091}{doi:\nolinkurl{10.1145/1640089.1640091}}


\bibitem[Panchekha et~al\mbox{.}(2019)]%
        {PanErnTatKam19Modular}
\bibfield{author}{\bibinfo{person}{Pavel Panchekha},
  \bibinfo{person}{Michael~D. Ernst}, \bibinfo{person}{Zachary Tatlock}, {and}
  \bibinfo{person}{Shoaib Kamil}.} \bibinfo{year}{2019}\natexlab{}.
\newblock \showarticletitle{Modular Verification of Web Page Layout}.
\newblock \bibinfo{journal}{\emph{Proc. ACM Program. Lang.}}
  \bibinfo{volume}{3}, \bibinfo{number}{OOPSLA}, Article
  \bibinfo{articleno}{151} (\bibinfo{date}{Oct.} \bibinfo{year}{2019}),
  \bibinfo{numpages}{26}~pages.
\newblock
\href{https://doi.org/10.1145/3360577}{doi:\nolinkurl{10.1145/3360577}}


\bibitem[Panchekha et~al\mbox{.}(2018)]%
        {PanGelErnTatKam18Verifying}
\bibfield{author}{\bibinfo{person}{Pavel Panchekha}, \bibinfo{person}{Adam~T.
  Geller}, \bibinfo{person}{Michael~D. Ernst}, \bibinfo{person}{Zachary
  Tatlock}, {and} \bibinfo{person}{Shoaib Kamil}.}
  \bibinfo{year}{2018}\natexlab{}.
\newblock \showarticletitle{Verifying That Web Pages Have Accessible Layout}.
  In \bibinfo{booktitle}{\emph{Proceedings of the 39th ACM SIGPLAN Conference
  on Programming Language Design and Implementation}} (Philadelphia, PA, USA)
  \emph{(\bibinfo{series}{PLDI '18})}. \bibinfo{publisher}{Association for
  Computing Machinery}, \bibinfo{address}{New York, NY, USA},
  \bibinfo{pages}{1–14}.
\newblock
\showISBNx{9781450356985}
\href{https://doi.org/10.1145/3192366.3192407}{doi:\nolinkurl{10.1145/3192366.3192407}}


\bibitem[Panchekha and Torlak(2016)]%
        {PanTor16Automated}
\bibfield{author}{\bibinfo{person}{Pavel Panchekha} {and}
  \bibinfo{person}{Emina Torlak}.} \bibinfo{year}{2016}\natexlab{}.
\newblock \showarticletitle{Automated Reasoning for Web Page Layout}. In
  \bibinfo{booktitle}{\emph{Proceedings of the 2016 ACM SIGPLAN International
  Conference on Object-Oriented Programming, Systems, Languages, and
  Applications}} (Amsterdam, Netherlands) \emph{(\bibinfo{series}{OOPSLA
  '16})}. \bibinfo{publisher}{Association for Computing Machinery},
  \bibinfo{address}{New York, NY, USA}, \bibinfo{pages}{181–194}.
\newblock
\showISBNx{9781450344449}
\href{https://doi.org/10.1145/2983990.2984010}{doi:\nolinkurl{10.1145/2983990.2984010}}


\bibitem[Park et~al\mbox{.}(2015)]%
        {ParSteAndRos15KJS}
\bibfield{author}{\bibinfo{person}{Daejun Park}, \bibinfo{person}{Andrei
  Stef\u{a}nescu}, {and} \bibinfo{person}{Grigore Ro\c{s}u}.}
  \bibinfo{year}{2015}\natexlab{}.
\newblock \showarticletitle{KJS: A Complete Formal Semantics of JavaScript}. In
  \bibinfo{booktitle}{\emph{Proceedings of the 36th ACM SIGPLAN Conference on
  Programming Language Design and Implementation}} (Portland, OR, USA)
  \emph{(\bibinfo{series}{PLDI '15})}. \bibinfo{publisher}{Association for
  Computing Machinery}, \bibinfo{address}{New York, NY, USA},
  \bibinfo{pages}{346–356}.
\newblock
\showISBNx{9781450334686}
\href{https://doi.org/10.1145/2737924.2737991}{doi:\nolinkurl{10.1145/2737924.2737991}}


\bibitem[Park et~al\mbox{.}(2021)]%
        {ParParAnRyu21JISET}
\bibfield{author}{\bibinfo{person}{Jihyeok Park}, \bibinfo{person}{Jihee Park},
  \bibinfo{person}{Seungmin An}, {and} \bibinfo{person}{Sukyoung Ryu}.}
  \bibinfo{year}{2021}\natexlab{}.
\newblock \showarticletitle{JISET: JavaScript IR-based Semantics Extraction
  Toolchain}. In \bibinfo{booktitle}{\emph{Proceedings of the 35th IEEE/ACM
  International Conference on Automated Software Engineering}} (Virtual Event,
  Australia) \emph{(\bibinfo{series}{ASE '20})}.
  \bibinfo{publisher}{Association for Computing Machinery},
  \bibinfo{address}{New York, NY, USA}, \bibinfo{pages}{647–658}.
\newblock
\showISBNx{9781450367684}
\href{https://doi.org/10.1145/3324884.3416632}{doi:\nolinkurl{10.1145/3324884.3416632}}


\bibitem[Politz et~al\mbox{.}(2012)]%
        {PolCarLerPomKri12Tested}
\bibfield{author}{\bibinfo{person}{Joe~Gibbs Politz},
  \bibinfo{person}{Matthew~J. Carroll}, \bibinfo{person}{Benjamin~S. Lerner},
  \bibinfo{person}{Justin Pombrio}, {and} \bibinfo{person}{Shriram
  Krishnamurthi}.} \bibinfo{year}{2012}\natexlab{}.
\newblock \showarticletitle{A Tested Semantics for Getters, Setters, and Eval
  in JavaScript}. In \bibinfo{booktitle}{\emph{Proceedings of the 8th Symposium
  on Dynamic Languages}} (Tucson, Arizona, USA) \emph{(\bibinfo{series}{DLS
  '12})}. \bibinfo{publisher}{Association for Computing Machinery},
  \bibinfo{address}{New York, NY, USA}, \bibinfo{pages}{1–16}.
\newblock
\showISBNx{9781450315647}
\href{https://doi.org/10.1145/2384577.2384579}{doi:\nolinkurl{10.1145/2384577.2384579}}


\bibitem[Pranjal(2019)]%
        {Pra19useState}
\bibfield{author}{\bibinfo{person}{Pranjal}.} \bibinfo{year}{2019}\natexlab{}.
\newblock \bibinfo{title}{The useState set method is not reflecting a change
  immediately}.
\newblock \bibinfo{howpublished}{\url{https://stackoverflow.com/q/54069253}}.
\newblock
\newblock
\shownote{Accessed: 2025-03-24}.


\bibitem[Radanne et~al\mbox{.}(2016a)]%
        {RadPapVouBal16Eliom}
\bibfield{author}{\bibinfo{person}{Gabriel Radanne}, \bibinfo{person}{Vasilis
  Papavasileiou}, \bibinfo{person}{J\'{e}r\^{o}me Vouillon}, {and}
  \bibinfo{person}{Vincent Balat}.} \bibinfo{year}{2016}\natexlab{a}.
\newblock \showarticletitle{Eliom: tierless Web programming from the ground
  up}. In \bibinfo{booktitle}{\emph{Proceedings of the 28th Symposium on the
  Implementation and Application of Functional Programming Languages}} (Leuven,
  Belgium) \emph{(\bibinfo{series}{IFL '16})}. \bibinfo{publisher}{Association
  for Computing Machinery}, \bibinfo{address}{New York, NY, USA}, Article
  \bibinfo{articleno}{8}, \bibinfo{numpages}{12}~pages.
\newblock
\showISBNx{9781450347679}
\href{https://doi.org/10.1145/3064899.3064901}{doi:\nolinkurl{10.1145/3064899.3064901}}


\bibitem[Radanne et~al\mbox{.}(2016b)]%
        {RadVouBal16Eliom}
\bibfield{author}{\bibinfo{person}{Gabriel Radanne},
  \bibinfo{person}{J{\'e}r{\^o}me Vouillon}, {and} \bibinfo{person}{Vincent
  Balat}.} \bibinfo{year}{2016}\natexlab{b}.
\newblock \showarticletitle{Eliom: A Core ML Language for Tierless Web
  Programming}. In \bibinfo{booktitle}{\emph{Programming Languages and
  Systems}}, \bibfield{editor}{\bibinfo{person}{Atsushi Igarashi}} (Ed.).
  \bibinfo{publisher}{Springer International Publishing},
  \bibinfo{address}{Cham}, \bibinfo{pages}{377--397}.
\newblock
\showISBNx{978-3-319-47958-3}


\bibitem[Rival and Yi(2020)]%
        {RivYi20Introduction}
\bibfield{author}{\bibinfo{person}{Xavier Rival} {and}
  \bibinfo{person}{Kwangkeun Yi}.} \bibinfo{year}{2020}\natexlab{}.
\newblock \bibinfo{booktitle}{\emph{Introduction to Static Analysis: An
  Abstract Interpretation Perspective}}.
\newblock \bibinfo{publisher}{The MIT Press}, \bibinfo{address}{Cambridge,
  Massachusetts}.
\newblock
\showISBNx{978-0-262-04341-0}


\bibitem[Ryu and Park(2024)]%
        {RyuPar24JavaScript}
\bibfield{author}{\bibinfo{person}{Sukyoung Ryu} {and} \bibinfo{person}{Jihyeok
  Park}.} \bibinfo{year}{2024}\natexlab{}.
\newblock \showarticletitle{JavaScript Language Design and Implementation in
  Tandem}.
\newblock \bibinfo{journal}{\emph{Commun. ACM}} \bibinfo{volume}{67},
  \bibinfo{number}{5} (\bibinfo{date}{May} \bibinfo{year}{2024}),
  \bibinfo{pages}{86–95}.
\newblock
\showISSN{0001-0782}
\href{https://doi.org/10.1145/3624723}{doi:\nolinkurl{10.1145/3624723}}


\bibitem[Serrano et~al\mbox{.}(2006)]%
        {SerGalLoi06Hop}
\bibfield{author}{\bibinfo{person}{Manuel Serrano}, \bibinfo{person}{Erick
  Gallesio}, {and} \bibinfo{person}{Florian Loitsch}.}
  \bibinfo{year}{2006}\natexlab{}.
\newblock \showarticletitle{Hop, a Language for Programming the Web 2.0}. In
  \bibinfo{booktitle}{\emph{Companion to the 21st ACM SIGPLAN Symposium on
  Object-Oriented Programming Systems, Languages, and Applications}} (Portland,
  Oregon, USA) \emph{(\bibinfo{series}{OOPSLA '06})}.
  \bibinfo{publisher}{Association for Computing Machinery},
  \bibinfo{address}{New York, NY, USA}, \bibinfo{pages}{975--985}.
\newblock


\bibitem[{Stack Exchange, Inc.}(2024)]%
        {Sta24Stack}
\bibfield{author}{\bibinfo{person}{{Stack Exchange, Inc.}}}
  \bibinfo{year}{2024}\natexlab{}.
\newblock \bibinfo{title}{2024 Stack Overflow Developer Survey}.
\newblock \bibinfo{howpublished}{\url{https://survey.stackoverflow.co/2024/}}.
\newblock
\newblock
\shownote{Accessed: 2025-01-21}.


\bibitem[Staff(2016)]%
        {Cac16React}
\bibfield{author}{\bibinfo{person}{CACM Staff}.}
  \bibinfo{year}{2016}\natexlab{}.
\newblock \showarticletitle{React: Facebook's Functional Turn on Writing
  JavaScript}.
\newblock \bibinfo{journal}{\emph{Commun. ACM}} \bibinfo{volume}{59},
  \bibinfo{number}{12} (\bibinfo{date}{Dec.} \bibinfo{year}{2016}),
  \bibinfo{pages}{56–62}.
\newblock
\showISSN{0001-0782}
\href{https://doi.org/10.1145/2980991}{doi:\nolinkurl{10.1145/2980991}}


\bibitem[Tang(2019)]%
        {Tan19React}
\bibfield{author}{\bibinfo{person}{Quoc~Van Tang}.}
  \bibinfo{year}{2019}\natexlab{}.
\newblock \bibinfo{title}{React hooks useState setValue still rerender one more
  time when value is equal}.
\newblock \bibinfo{howpublished}{\url{https://stackoverflow.com/q/57652176}}.
\newblock
\newblock
\shownote{Accessed: 2025-03-24}.


\bibitem[Tehila(2022)]%
        {Teh19Updating}
\bibfield{author}{\bibinfo{person}{Tehila}.} \bibinfo{year}{2022}\natexlab{}.
\newblock \bibinfo{title}{Updating state to the same value directly in the
  component body during render causes infinite loop}.
\newblock \bibinfo{howpublished}{\url{https://stackoverflow.com/q/74034014}}.
\newblock
\newblock
\shownote{Accessed: 2025-03-24}.


\bibitem[vadirn(2018)]%
        {vad18React}
\bibfield{author}{\bibinfo{person}{vadirn}.} \bibinfo{year}{2018}\natexlab{}.
\newblock \bibinfo{title}{Does React batch state update functions when using
  hooks?}
\newblock \bibinfo{howpublished}{\url{https://stackoverflow.com/q/53048495}}.
\newblock
\newblock
\shownote{Accessed: 2025-03-24}.


\bibitem[Veronese et~al\mbox{.}(2023)]%
        {VerFarBerTemSquMaf23WebSpec}
\bibfield{author}{\bibinfo{person}{Lorenzo Veronese}, \bibinfo{person}{Benjamin
  Farinier}, \bibinfo{person}{Pedro Bernardo}, \bibinfo{person}{Mauro
  Tempesta}, \bibinfo{person}{Marco Squarcina}, {and} \bibinfo{person}{Matteo
  Maffei}.} \bibinfo{year}{2023}\natexlab{}.
\newblock \showarticletitle{{{WebSpec}}: {{Towards Machine-Checked Analysis}}
  of {{Browser Security Mechanisms}}}. In \bibinfo{booktitle}{\emph{2023 {{IEEE
  Symposium}} on {{Security}} and {{Privacy}}}} \emph{(\bibinfo{series}{SP
  '23})}. \bibinfo{publisher}{IEEE}, \bibinfo{address}{San Francisco, CA, USA},
  \bibinfo{pages}{2761--2779}.
\newblock
\href{https://doi.org/10.1109/sp46215.2023.10179465}{doi:\nolinkurl{10.1109/sp46215.2023.10179465}}


\bibitem[Vouillon and Balat(2014)]%
        {VouBal14Bytecode}
\bibfield{author}{\bibinfo{person}{J\'{e}r\^{o}me Vouillon} {and}
  \bibinfo{person}{Vincent Balat}.} \bibinfo{year}{2014}\natexlab{}.
\newblock \showarticletitle{From bytecode to JavaScript: the Js\_of\_ocaml
  compiler}.
\newblock \bibinfo{journal}{\emph{Softw. Pract. Exper.}} \bibinfo{volume}{44},
  \bibinfo{number}{8} (\bibinfo{date}{Aug.} \bibinfo{year}{2014}),
  \bibinfo{pages}{951–972}.
\newblock
\showISSN{0038-0644}
\href{https://doi.org/10.1002/spe.2187}{doi:\nolinkurl{10.1002/spe.2187}}


\bibitem[Weisenburger et~al\mbox{.}(2020)]%
        {WeiWirSal21Survey}
\bibfield{author}{\bibinfo{person}{Pascal Weisenburger},
  \bibinfo{person}{Johannes Wirth}, {and} \bibinfo{person}{Guido Salvaneschi}.}
  \bibinfo{year}{2020}\natexlab{}.
\newblock \showarticletitle{A Survey of Multitier Programming}.
\newblock \bibinfo{journal}{\emph{ACM Comput. Surv.}} \bibinfo{volume}{53},
  \bibinfo{number}{4}, Article \bibinfo{articleno}{81} (\bibinfo{date}{Sept.}
  \bibinfo{year}{2020}), \bibinfo{numpages}{35}~pages.
\newblock
\showISSN{0360-0300}
\href{https://doi.org/10.1145/3397495}{doi:\nolinkurl{10.1145/3397495}}


\end{thebibliography}

\clearpage
\appendix
\section{Full Operational Semantics of \Rtrace}\label{app:sec:full}
\subsection{Semantic Domains}\label{app:sec:dom}
\begin{center}
  \begin{bnf}[rccll]
    v : $\dom{Val}$ ::=
    | $k$ : constant values
    | $cl$ : closure
    | $C$ : component name
    | $cs$ : component spec
    | $[\Overline{s}]$ : view spec list
    | $\<\ell, p\>$ : setter closure
    ;;
    k : $\dom{Const}$ ::= $\<\>$ // $\TT$ // $\FF$ // $n$ : constant values
    ;;
    cl : $\dom{Clos}$ ::=
    | $\clos{x}{e}{\sigma}$ : closure
    ;;
    \sigma : $\dom{Env}$ ::= $[\Overline[\ell]{x \mapsto v}]$ : environment
    ;;
    cs : $\dom{ComSpec}$ ::=
    | $\<C, v\>$ : component spec
    ;;
    s : $\dom{ViewSpec}$ ::= $k$ // $cl$ // $[\Overline{s}]$ // $cs$ : view spec
    ;;
    \delta : $\dom{DefTable}$ ::= $[\Overline[\ell]{C \mapsto \deftabent{x}{e}}]$ : component def table
    ;;
    \ell :in: $\mathbb{N}$ : label
    ;;
    p :in: $\dom{Path}$ : tree path
    ;;
    m : $\dom{TreeMem}$ ::= $[\Overline[\ell]{p \mapsto \pi}]$ : tree memory
    ;;
    \pi : $\dom{View}$ ::= $\view{cs}{\{\Overline{d}\}}{\rho}{q}{t}$ : state environment
    ;;
    d : $\dom{Decision}$ ::= \dec{Check} // \dec{Effect} : decision
    ;;
    t : $\dom{Tree}$ ::= $k$ // $cl$ // $[\Overline{t}]$ // $p$ : tree
    ;;
    \rho : $\dom{SttStore}$ ::= $[\Overline[\ell]{\ell \mapsto \sttstent{v}{q}}]$ : state store
    ;;
    q : $\dom{JobQ}$ ::= $[\Overline[\ell]{cl}]$ : job queue
    ;;
    \Sigma : $\dom{Context}$ ::= $m$ // $\pi$ : evaluation context
    ;;
    \phi : $\dom{Phase}$ ::= \phase{Init} // \phase{Succ} // \phase{Normal} : phase
    ;;
    \mu : $\dom{Mode}$ ::= $\rendermode$ \textrm{(rendered)} // $\checkmode$ \textrm{(check)} // $\eloopmode$  \textrm{(event loop)} : runtime mode
  \end{bnf}
\end{center}

\subsection{Operational Semantics}\label{app:subsec:sem}
\begin{center}
  \hfill\fbox{$\<e, \delta\> \text{ or }\<t, m, \omega, \delta, \mu\> \smallstep \<t, m', \omega', \delta, \mu'\>$}
\begin{mathpar}
\inferrule[StepInit]{
  \Step<\phase{Normal}><->[[], []]{e}{s, [], \omega} \\
  [] \vdash \sem{init}(s) = \<t, m, \omega'\>
}{
  \<e, \delta\> \smallstep \<t, m, \omega \concat \omega', \delta, \rendermode\>
} \and
\inferrule[StepEffect]{
  m \vdash \sem{commitEffs}(t) = \<m', \omega'\>
}{
  \<t, m, \omega, \delta, \rendermode\> \smallstep \<t, m', \omega \concat \omega', \delta, \checkmode\>
} \and
\inferrule[StepCheck]{
  m, \delta \vdash \sem{check}(t) = \<\mu, m', \omega'\> \\
}{
  \<t, m, \omega, \delta, \checkmode\> \smallstep \<t, m', \omega \concat \omega', \delta, \mu\>
} \and
\inferrule[StepEvent]{
  \clos{x}{e}{\sigma} \in \sem{handlers}(m, t) \\
  \Step<\phase{Normal}><->[m, \sigma[x \mapsto \<\>]]{e}{v, m', \omega'} \\
}{
  \<t, m, \omega, \delta, \eloopmode\> \smallstep \<t, m', \omega \concat \omega', \delta, \checkmode\>
}
\end{mathpar}
\end{center}

\begin{center}
  \hfill\fbox{$\sem{handlers}(m, t) = \{\Overline[\ell]{cl}\}$}
  \[
    \sem{handlers}(m, t) = \begin{dcases*}
      \{\} & if $t = k$ \\
      \{ cl \} & if $t = cl$ \\
      \bigcup_{i = 1}^n \sem{handlers}(m, t_i) & if $t = [\Overline{t_i}]_{i = 1}^n$ \\
      \sem{handlers}(m, m[p].\field{children}) & if $t = p$
    \end{dcases*}
  \]
\end{center}

\begin{center}
  \hfill\fbox{$\Step<\phi><p>[\Sigma, \sigma]{e}{v, \Sigma', \omega}$}
  \begin{mathparpagebreakable}
    \inferrule[Unit]{ }{\Step{\texttt{()}}{\<\>, \Sigma}, []} \and
    \inferrule[True]{ }{\Step{\texttt{true}}{\TT, \Sigma}, []} \and
    \inferrule[False]{ }{\Step{\texttt{false}}{\FF, \Sigma}, []} \and
    \inferrule[Int]{ }{\Step{n}{n, \Sigma}, []} \and
    \inferrule[Var]{ }{\Step{x}{\sigma(x), \Sigma}, []} \and
    \inferrule[Bop]{
      \Step[\Sigma, \sigma]{e_1}{n_1, \Sigma_1, \omega_1} \\
      \Step[\Sigma_1, \sigma]{e_2}{n_2, \Sigma_2, \omega_2} \\
      \text{$f_\oplus \in \mathbb{Z} \times \mathbb{Z} \to \mathbb{Z}$}
    }{\Step[\Sigma, \sigma]{e_1 \oplus e_2}{f_\oplus(n_1, n_2), \Sigma_2, \omega_1 \concat \omega_2}} \and
    \inferrule[Cond]{
      \Step[\Sigma, \sigma]{e_1}{b, \Sigma', \omega} \\
      \Step[\Sigma', \sigma]{\scond{b}{e_2}{e_3}}{v, \Sigma'', \omega'}
    }{
      \Step[\Sigma, \sigma]{\cond{e_1}{e_2}{e_3}}{v, \Sigma'', \omega \concat \omega'}
    } \and
    \inferrule[Func]{ }{\Step[\Sigma, \sigma]{\func{x}{e}}{\clos{x}{e}{\sigma}, \Sigma, []}} \and
    \inferrule[Seq]{
      \Step[\Sigma, \sigma]{e_1}{\_, \Sigma_1, \omega_1} \\
      \Step[\Sigma_1, \sigma]{e_2}{v_2, \Sigma_2, \omega_2}
    }{\Step[\Sigma, \sigma]{\seq{e_1}{e_2}}{v_2, \Sigma_2, \omega_1 \concat \omega_2}} \and
    \inferrule[List]{
      \bigl(\Step[\Sigma_{i-1}, \sigma]{e_i}{s_i, \Sigma_i, \omega_i}\bigr)_{i = 1}^n
    }{
      \Step[\Sigma_0, \sigma]{
        [\Overline{e_i}]_{i = 1}^n}{[\Overline{s_i}]_{i = 1}^n, \Sigma_n, {\textstyle\Concat_{i = 1}^n \omega_i}
      }
    } \and
    \inferrule[LetBind]{
      \Step{e_1}{v_1, \Sigma_1, \omega_1} \\
      \Step[\Sigma_1, \sigma[x \mapsto v_1]]{e_2}{v_2, \Sigma_2, \omega_2}
    }{
      \Step{\letbind{x}{e_1}{e_2}}{v_2, \Sigma_2, \omega_1 \concat \omega_2}
    } \and
    \inferrule[AppFunc]{
  \Step{e_1}{\clos{x}{e}{\sigma'}, \Sigma_1, \omega_1} \\
  \Step[\Sigma_1, \sigma]{e_2}{v_2, \Sigma_2, \omega_2} \\
  \Step[\Sigma_2, \sigma'[x \mapsto v_2]]{e}{v', \Sigma', \omega'}
}{
  \Step{\app{e_1}{e_2}}{v', \Sigma', \omega_1 \concat \omega_2 \concat \omega'}
} \and
\inferrule[AppCom]{
  \Step{e_1}{C, \Sigma_1, \omega_1} \\
  \Step[\Sigma_1, \sigma]{e_2}{v, \Sigma_2, \omega_2}
}{
  \Step{\app{e_1}{e_2}}{\<C, v\>, \Sigma_2, \omega_1 \concat \omega_2}
} \and
\inferrule[AppSetComp]{
  \Step[\pi, \sigma]{e_1}{\<\ell, p\>, \pi_1, \omega_1} \\
  \Step[\pi_1, \sigma]{e_2}{cl, \pi_2, \omega_2} \\
  \phi \in \{\phase{Init}, \phase{Succ}\}
}{
  \Step[\pi, \sigma]{\app{e_1}{e_2}}{
    \<\>,
    \pi_2\!\left\{
      \begin{spreadlines}{0pt}
        \begin{lgathered}
          \text{let $\rho = \pi_2.\field{sttst},\ vq = \rho[\ell]$ in} \\
          \begin{aligned}
            \field{dec} &\colon \pi_2.\field{dec} \cup \{\dec{Check}\} \\
            \field{sttst} &\colon \rho[\ell \mapsto vq\{\fv{sttq}{vq.\field{sttq} \concat [cl]}\}]
          \end{aligned}
        \end{lgathered}
      \end{spreadlines}
    \right\},
    \omega_1 \concat \omega_2
  }
} \and
\inferrule[AppSetNormal]{
  \Step<\phase{Normal}><->[m, \sigma]{e_1}{\<\ell, p\>, m_1, \omega_1} \\
  \Step<\phase{Normal}><->[m_1, \sigma]{e_2}{cl, m_2, \omega_2} \\
}{
  \Step<\phase{Normal}><->[m, \sigma]{\app{e_1}{e_2}}{
    \<\>,
    m_2\!\left[
      \begin{spreadlines}{0pt}
        \begin{lgathered}
          \text{let $\pi = m_2[p],\ \rho = \pi.\field{sttst},\ vq = \rho[\ell]$ in} \\
          p \mapsto \pi\!\left\{
            \begin{aligned}
              \field{dec} &\colon \pi.\field{dec} \cup \{\dec{Check}\} \\
              \field{sttst} &\colon \rho[\ell \mapsto vq\{\fv{sttq}{vq.\field{sttq} \concat [cl]}\}]
            \end{aligned}
          \right\}
        \end{lgathered}
      \end{spreadlines}
    \right],
    \omega_1 \concat \omega_2
  }
} \and
\inferrule[SttBind]{
  \Step<\phase{Init}>[\pi, \sigma]{e_1}{v_1, \pi_1, \omega_1} \\
  \Step<\phase{Init}>[{
    \begin{spreadlines}{0pt}
      \pi_1\!\left\{
        \fv{sttst}{
          \pi_1.\field{sttst}\!\left[\ell \mapsto \sttstent*{v_1}{[]}\right]
        }
      \right\}\!,
      \sigma\!\left[
        \begin{aligned}
          x &\mapsto v_1 \\
          x_{\textsf{set}} &\mapsto \<\ell, p\>
        \end{aligned}
      \right]
    \end{spreadlines}
  }]{e_2}{v_2, \pi_2, \omega_2}
}{
  \Step<\phase{Init}>[\pi, \sigma]{\stbind[\ell]{x}{e_1}{e_2}}{v_2, \pi_2, \omega_1 \concat \omega_2}
} \and
\inferrule[SttReBind]{
  \pi_0.\field{sttst}[\ell] = \bigl\{
    \fv{val}{v_0},\:\fv{sttq}{[\Overline{\clos{x_i}{e'_i}{\sigma_i}}]_{i = 1}^n}
  \bigr\} \\
  \bigl(\Step[\pi_{i-1}, \sigma_i[x_i \mapsto v_{i-1}]]{e'_i}{v_i, \pi_i, \omega_i}\bigr)_{i = 1}^n \\
  \Step[{
    \begin{spreadlines}{0pt}
      \pi_n\!\left\{
        \begin{aligned}
          \field{dec} &\colon \pi_n.\field{dec} \cup \scond{v_n \not\equiv v_0 }{\{\dec{Effect}\}}{\emptyset} \\
          \field{sttst} &\colon \pi_n.\field{sttst}[\ell \mapsto \sttstent{v_n}{[]}]
        \end{aligned}
      \right\}\!,
      \sigma\!\left[
        \begin{aligned}
          x &\mapsto v_n \\
          x_{\textsf{set}} &\mapsto \<\ell, p\>
        \end{aligned}
      \right]
    \end{spreadlines}
  }]{e_2}{v, \pi', \omega'}
}{
  \Step<\phase{Succ}>[\pi_0, \sigma]{\stbind{x}{e_1}{e_2}}{v, \pi', ({\textstyle\Concat_{i = 0}^n \omega_i}) \concat \omega'}
} \and
\inferrule[Eff]{
  \phi \in \{\phase{Init}, \phase{Succ}\}
}{
  \Step[\pi, \sigma]{\eff{e}}{
    \begin{spreadlines}{0pt}
      \<\>,
      \pi\{ \fv{effq}{\pi.\field{effq} \concat [\clos{\_}{e}{\sigma}]} \},
      []
    \end{spreadlines}
  }
} \and
\inferrule[Print]{
  \Step{e}{v, \Sigma', \omega}
}{
  \Step{\print{e}}{\<\>, \Sigma', \omega \concat [v]}
}
  \end{mathparpagebreakable}
\end{center}

\begin{center}
  \hfill\fbox{$\Step*<\phi><p>[\pi, \sigma]{e}{s, \pi', \omega}$}
\begin{mathpar}
  \inferrule[EvalOnce]{
    \Step[\pi\{\fv{dec}{\pi.\field{dec}\setminus\{\dec{Check}\}},\fv{effq}{[]}\}, \sigma]{e}{s, \pi', \omega} \\
    \dec{Check} \notin \pi'.\field{dec} \\
    \phi \in \{\phase{Init}, \phase{Succ}\}
  }{
    \Step*[\pi, \sigma]{e}{s, \pi', \omega}
  } \and
  \inferrule[EvalMult]{
    \Step[\pi\{\fv{dec}{\pi.\field{dec}\setminus\{\dec{Check}\}},\fv{effq}{[]}\}, \sigma]{e}{s, \pi', \omega} \\
    \dec{Check} \in \pi'.\field{dec} \\
    \Step*<\phase{Succ}>[\pi', \sigma]{e}{s', \pi'', \omega'} \\
    \phi \in \{\phase{Init}, \phase{Succ}\}
  }{
    \Step*[\pi, \sigma]{e}{s', \pi'', \omega \concat \omega'}
  }
\end{mathpar}
\end{center}

\begin{center}
  \hfill\fbox{$m, \delta \vdash \sem{init}(s) = \<t, m', \omega\>$}
\begin{mathpar}
  \inferrule[InitConst]{ }{m, \delta \vdash \sem{init}(k) = \<k, m, []\>} \and
  \inferrule[InitClos]{ }{m, \delta \vdash \sem{init}(cl) = \<cl, m, []\>} \and
  \inferrule[InitArray]{
    \bigl(m_{i-1}, \delta \vdash \sem{init}(s_i) = \<t_i, m_i, \omega_i\>\bigr)_{i = 1}^n
  }{
    m_0, \delta \vdash \sem{init}([\Overline{s_i}]_{i = 1}^n) = \<[\Overline{t_i}]_{i = 1}^n, m_n, {\textstyle\Concat_{i = 1}^n \omega_i}\>
  } \and
  \inferrule[InitCom]{
    m \vdash p\:\judge{fresh} \\
    \delta[C] = \deftabent{x}{e} \\
    \Step*<\phase{Init}>[
      \view{\<C, v\>}{\emptyset}{[]}{[]}{\<\>},
      [x \mapsto v]
    ]{e}{s, \pi, \omega} \\
    m[p \mapsto \pi], \delta \vdash \sem{init}(s) = \<t, m', \omega'\>
  }{
    m, \delta \vdash \sem{init}(\<C, v\>) = \bigl\<p, m'[p \mapsto \pi\{\fv{dec}{\{\dec{Effect}\}}, \fv{child}{t}\}], \omega \concat \omega'\bigr\>
  }
\end{mathpar}
\end{center}

\begin{center}
  \hfill\fbox{$m, \delta \vdash \sem{check}(t) = \<\mu, m', \omega\>$}
\begin{mathpar}
  \inferrule[CheckConst]{ }{m, \delta \vdash \sem{check}(k) = \<\eloopmode, m, []\>} \and
  \inferrule[CheckClos]{ }{m, \delta \vdash \sem{check}(cl) = \<\eloopmode, m, []\>} \and
  \inferrule[CheckArray]{
    \bigl(m_{i-1}, \delta \vdash \sem{check}(t_i) = \<\mu_i, m_i, \omega_i\>\bigr)_{i = 1}^n
  }{\textstyle
    m_0, \delta \vdash \sem{check}([\Overline{t_i}]_{i = 1}^n) = \<\bigsqcup_{i = 1}^n \mu_i, m_n, {\textstyle\Concat_{i = 1}^n \omega_i}\>
  } \and
  \inferrule[CheckIdle]{
    m[p] = \pi \\
    \dec{Check} \notin \pi.\field{dec} \\
    m, \delta \vdash \sem{check}(\pi.\field{child}) = \<\mu, m', \omega\>
  }{\textstyle
    m, \delta \vdash \sem{check}(p) = \<\mu, m', \omega\>
  } \and
  \inferrule[CheckNoEffect]{
    m[p] = \pi \\
    \dec{Check} \in \pi.\field{dec} \\
    \pi.\field{spec} = \<C, v\> \\
    \delta[C] = \deftabent{x}{e} \\
    \Step*<\phase{Succ}>[
      \pi,
      [x \mapsto v]
    ]{e}{\_, \pi', \omega} \\
    \dec{Effect} \notin \pi'.\field{dec} \\
    m, \delta \vdash \sem{check}(\pi.\field{child}) = \<\mu, m', \omega'\>
  }{\textstyle
    m, \delta \vdash \sem{check}(p) = \<\mu, m'[p \mapsto \pi'], \omega \concat \omega'\>
  } \and
  \inferrule[CheckEffect]{
    m[p] = \pi \\
    \dec{Check} \in \pi.\field{dec} \\
    \pi.\field{spec} = \<C, v\> \\
    \delta[C] = \deftabent{x}{e} \\
    \Step*<\phase{Succ}>[
      \pi,
      [x \mapsto v]
    ]{e}{s, \pi', \omega} \\
    \dec{Effect} \in \pi'.\field{dec} \\
    m, \delta \vdash \sem{reconcile}(\pi.\field{child}, s) = \<t', m', \omega'\> \\
  }{\textstyle
    m, \delta \vdash \sem{check}(p) = \<\rendermode, m'[p \mapsto \pi'\{\fv{child}{t'}\}], \omega \concat \omega'\>
  }
\end{mathpar}
\end{center}

\begin{center}
  \hfill\fbox{$m, \delta \vdash \sem{reconcile}(t, s) = \<t', m', \omega\>$}
\begin{mathpar}
  \inferrule[ReconcileArray]{
    \bigl(m_{i-1}, \delta \vdash \sem{reconcile}(t_i, s_i) = \<t_i', m_i, \omega_i\>\bigr)_{i = 1}^n
  }{
    m_0, \delta
    \vdash \sem{reconcile}([\Overline{t_i}]_{i = 1}^n, [\Overline{s_i}]_{i = 1}^n)
    = \<[\Overline{t_i'}]_{i = 1}^n, m_n, {\textstyle\Concat_{i = 1}^n \omega_i}\>
  } \and
  \inferrule[ReconcileComEffect]{
    m[p] = \pi \\
    \pi.\field{spec} = \<C, \_\> \\
    \delta[C] = \deftabent{x}{e} \\
    \Step*<\phase{Succ}>[
      \pi\{\fv{spec}{\<C, v\>}\},
      [x \mapsto v]
    ]{e}{s, \pi', \omega} \\
    m, \delta \vdash \sem{reconcile}(\pi.\field{child}, s) = \<t', m', \omega'\> \\
  }{
    m, \delta \vdash \sem{reconcile}(p, \<C, v\>) = \bigl\<p, m'[p \mapsto \pi'\{\fv{dec}{\{\dec{Effect}\}}, \fv{child}{t'}\}], \omega \concat \omega'\bigr\>
  } \and
  \inferrule[ReconcileComNew]{
    m[p].\field{spec} = \<C', \_\> \\
    C \ne C' \\
    m, \delta \vdash \sem{init}(\<C, v\>) = \<t', m', \omega\>
  }{
    m, \delta \vdash \sem{reconcile}(p, \<C, v\>) = \<t', m', \omega\>
  } \and
  \inferrule[ReconcileOther]{
    \<t, s\> \neq \<[\Overline{t}], [\Overline{s}]\> \\
    \<t, s\> \neq \<p, \<C, v\>\> \\
    m, \delta \vdash \sem{init}(s) = \<t', m', \omega\>
  }{
    m, \delta \vdash \sem{reconcile}(t, s) = \<t', m', \omega\>
  }
\end{mathpar}
\end{center}

\begin{center}
  \hfill\fbox{$m \vdash \sem{commitEffs}(t) = \<m', \omega\>$}
\begin{mathpar}
  \inferrule[CommitEffsConst]{ }{m \vdash \sem{commitEffs}(k) = \<m, []\>} \and
  \inferrule[CommitEffsClos]{ }{m \vdash \sem{commitEffs}(cl) = \<m, []\>} \and
  \inferrule[CommitEffsArray]{
    \bigl(m_{i-1} \vdash \sem{commitEffs}(t_i) = m_i, \omega_i\bigr)_{i = 1}^n
  }{
    m_0 \vdash \sem{commitEffs}([\Overline{t_i}]_{i = 1}^n) = \<m_n, {\textstyle\Concat_{i = 1}^n \omega_i}\>
  } \and
  \inferrule[CommitEffsPathIdle]{
    \dec{Effect} \notin m[p].\field{dec} \\
    m \vdash \sem{commitEffs}(m[p].\field{child}) = \<m', \omega\> \\
  }{
    m \vdash \sem{commitEffs}(p) = \<m', \omega\>
  } \and
  \inferrule[CommitEffsPath]{
    \dec{Effect} \in m[p].\field{dec} \\
    m[p].\field{effq} = [\Overline{\clos{\_}{e_i}{\sigma_i}}]_{i = 1}^n \\
    m \vdash \sem{commitEffs}(m[p].\field{child}) = \<m_0, \omega_0\> \\
    \bigl(
      \Step<\phase{Normal}><->[m_{i-1}, \sigma_i]{e_i}{v_i, m_i, \omega_i}
    \bigr)_{i = 1}^n
  }{
    m \vdash \sem{commitEffs}(p) = \bigl\<m_n\bigl[p \mapsto m_n[p]\{\fv{dec}{m_n[p].\field{dec} \setminus \{\dec{Effect}\}}\}\bigr], {\textstyle\Concat_{i = 0}^n \omega_i}\bigr\>
  }
\end{mathpar}
\end{center}

\clearpage
\section{Proofs of Theorems}\label{app:sec:proofs}
\begin{defn}[$e$-Equivalent Views]\label{def:equiv}
  Views~$\pi$ and~$\pi'$ are $e$-equivalent when the state entries whose labels appear in~$e$ are the same.
  That is,
  \[
    \pi \equiv_e \pi'
    \triangleiff
    \forall \ell \in \sem{labels}(e),\ \pi.\field{sttst}[\ell] = \pi'.\field{sttst}[\ell]. \qedhere
  \]
\end{defn}

\begin{defn}[$t$-Equivalent Memories]\label{def:memequiv}
  Tree memories~$m$ and~$m'$ are $t$-equivalent iff the descendant views of~$t$ in~$m$ and~$m'$ are the same.
  That is,
  \[
    m \equiv_t m'
    \triangleiff
    \forall p \in \sem{reachable}(m,t), m[p] = m'[p] \qedhere
  \]
\end{defn}

\begin{lem}[Similar Evaluations]\label{lem:simeval}
  Evaluations of~$e$ in~\phase{Succ} phase with similar views as a context produce identical values and output buffers, along with $e$-equivalent views.
  That is, if $\Step<\phase{Succ}>[\pi,\sigma]{e}{v,\pi',\omega}$,
  then for any $\hat\pi \approx \pi$, we have
  $\Step<\phase{Succ}>[\hat \pi,\sigma]{e}{v,\hat \pi',\omega}$ where $\hat\pi' \equiv_e \pi'$.
\end{lem}

\begin{proof}
  Suppose we have $\Step<\phase{Succ}>[\pi,\sigma]{e}{v,\pi',\omega}$ and choose some $\hat\pi \approx \pi$.
  We need to show that $\hat v = v$, $\hat\pi' \equiv_e \pi'$, and $\hat\omega = \omega$, where $\Step<\phase{Succ}>[\hat\pi,\sigma]{e}{\hat v,\hat\pi',\hat\omega}$.

  We prove that $\hat v = v$, $\hat\pi' \approx \pi'$ (before showing $\hat\pi' \equiv_e \pi'$), and $\hat\omega = \omega$ by rule induction.
  We put a~$\hat\cdot$ on the intermediate variables in the derivation of~$\Step<\phase{Succ}>[\hat\pi,\sigma]{e}{\hat v,\hat\pi',\hat\omega}$.
  \begin{description}
    \item[{\normalfont\Rule{Print}}]
      Evaluations of~$e$ produce same values~$\hat v = v$ and output buffers~$\hat \omega = \omega$ with similar views~$\hat\pi' \approx \pi'$ by the IH.
      Appending the same values to the same output buffers results in the same buffers~$\hat\omega \concat [\hat v] = \omega \concat [v]$.
      Resulting values are both~$\<\>$.
      Contexts---views in this case---are unmodified.

    \item[{\normalfont\Rule{AppSetComp}}]
      Evaluations of~$e_1$ and~$e_2$ produce same values~$\<\hat\ell, \hat p\> = \<\ell, p\>$ and $\hat{cl} = cl$, and output buffers~$\hat\omega_1 = \omega_1$ and $\hat\omega_2 = \omega_2$, and similar views~$\hat\pi_1 \approx \pi_1$ and $\hat\pi_2 \approx \pi_2$ by the IH.
      Concatenations of the same output buffers result in the same buffers~$\hat\omega_1 \concat \hat\omega_2 = \omega_1 \concat \omega_2$.
      Appending the same closures~$\hat{cl} = cl$ to the state update queues of similar views~$\hat\pi_2 \approx \pi_2$ produces similar views.
      (The added \dec{Check}s are discarded anyway during normalization.)
      Resulting values are both~$\<\>$.

    \item[{\normalfont\Rule{SttReBind}}]
      Evaluations of the queued~$e_i'$s all produce the intermediate values~$\hat v_i = v_i$ and output buffers~$\hat \omega_i = \omega_i$ that are equivalent and views~$\hat\pi_i \approx \pi_i$ that are similar by the IH.
      Then the modified~$\hat\pi_n$ and~$\pi_n$ each modified with the added \dec{Effect} and the same state~$\hat v_n = v_n$ are still similar.
      This suffices to apply the IH to the evaluation of~$e_2$, resulting in the same values~$\hat v = v$ and output buffers~$\hat\omega' = \omega'$, and similar views~$\hat\pi' \approx \pi'$.
      Thus the resulting values~$\hat v = v$ and the concatenation of the buffers~$(\concat_{i=1}^n \hat \omega_i) \concat \hat\omega' = (\concat_{i=1}^n \omega_i) \concat \omega'$ are equivalent, and the resulting views~$\hat\pi' \approx \pi'$ are similar.
  \end{description}
  The remaining cases are trivial.

  It is easy to show that if~$\pi$ and~$\hat \pi$ are similar and $e$-equivalent for some~$e$, then the results are also $e$-equivalent using rule induction.

  Now we can prove that $\hat\pi' \equiv_e \pi'$, again using rule induction.
  In each case, evaluation of each sub-expression~$e'$ of~$e$ produces $e'$-equivalent views by the IH, and the equivalence is preserved in following evaluations of other sub-expressions.
  Therefore, the resulting views from~$\pi$ and~$\hat \pi$ are $e'$-equivalent for every sub-expression~$e'$, which makes the resulting~$\pi'$ and~$\hat \pi'$ $e$-equivalent.
  Note that nested or dead sub-expressions of~$e$ do not contain state labels~$\ell$, due to the syntactic restriction that Hooks can only appear at the top level of a component body.
  
  The only noteworthy case is \Rule{SttReBind} where a state label exists.
  In this case, the state store entry for~$\ell$ is updated directly to the same value~$v_n$ (as just shown above) in both derivations from~$\pi$ and~$\hat \pi$.
  Therefore, the resulting~$\pi'$ and~$\hat \pi'$ have the same state store entry for~$\ell$.
\end{proof}

\simevallem*

\begin{proof}
  Suppose $\Step*<\phase{Succ}>[\pi,\sigma[x\mapsto v]]{e}{v,\pi',\omega}$ and $\Step*<\phase{Succ}>[\hat \pi,\sigma[x\mapsto v]]{e}{\hat v',\hat \pi', \hat \omega}$, for some~$\pi$ such that~$\pi.\field{spec} = \<C, v\>$ where~$\delta[C] = \<\lfun{x}{e}, \sigma\>$ and~$\hat\pi \approx \pi$.
  By \cref{lem:simeval}, evaluations of~$e$ produce $e$-equivalent views, that is, the states whose labels appear in~$e$ are the same in~$\pi'$ and~$\hat \pi'$.
  Since~$\pi'$ and~$\hat \pi'$ are valid under~$\delta$, the only labels that~$\pi'$ and~$\hat \pi'$ have are those appearing in~$e$.
  Therefore, $\pi' = \hat \pi'$.
  Similarly, $v' = \hat v'$ and $\omega = \hat \omega$.
  Any remaining loops, if present, are evaluated under identical views, trivially producing the same results.
\end{proof}

\begin{defn}[Stability and Semi-Stability]\label{def:stability}
  A view~$\pi$ is \emph{stable} iff the re-evaluation of~$\pi$'s body (as in \Rule{EvalOnce} and \Rule{EvalMult}) produces the same view as $\pi$:
  \begin{multline*}
    \Step[\pi\{\fv{dec}{\pi.\field{dec}\setminus\{\dec{Check}\}},\fv{effq}{[]}\}, \sigma[x\mapsto v]]{e}{\_, \pi, \omega} \\
    \text{where $\pi.\field{spec} = \<C,v\>$ and $\delta[C] = \<\lambda x.e,\sigma\>$.}
  \end{multline*}

  View~$\pi$ is \emph{semi-stable} iff $\pi$ with empty state update queues is stable.

  We extend this notion to tree memory:
  $m$ is~\emph{$t$-stable} iff the descendant views of~$t$ in~$m$ are all stable, and \emph{$t$-semi-stable} iff the descendant views of~$t$ in~$m$ are all semi-stable.
  
  Since component body of a stable view can be re-evaluated in an idempotent way, evaluation of the body can be thought of as ``reading'' the specification of the view.
\end{defn}

\begin{defn}[Coherence]\label{def:coherence}
  A view is \emph{coherent} if its decision is coherent with the status of the state update queue.
  That is, a coherent view~$\pi$ satisfies the predicate
  \[
    \dec{Check} \in \pi.\field{dec}
    \qquad\longleftrightarrow\qquad
    \exists \ell \in \domain \rho,\ \rho[\ell].\field{sttq} \neq []
    \quad \text{where $\rho = \pi.\field{sttst}$.} 
  \]

  We extend this notion to tree memory, i.e., $m$ is \emph{$t$-coherent} iff the descendants of~$t$ in~$m$ are coherent.
\end{defn}

\begin{restatable}[Preservation of Validity]{lem}{validprop}\label{lem:valid}
  If $\<e,\delta\> \smallstep^* \<t,m,\omega,\delta,\mu\>$, then~$m$ is valid under~$\delta$.
\end{restatable}

\begin{proof}
  The only rule that adds states is \Rule{SttBind},
  which appears only in the derivation of $\sem{init}$ where the context being a fresh view.
  Therefore, it is sufficient to check that the evaluation of the component body in $\phase{Init}$ produces a valid view.
  Due to the syntactic restriction on the usage of Hooks, every |useState| Hook is executed during the evaluation, adding its state labeled~$\ell$ to the view.
  This process adds every label in the expression to the view, producing a valid view.
\end{proof}

\begin{lem}[Semi-Stability and Normalized Form]\label{lem:stabnorm}
  Let~$\pi$ be a semi-stable view valid under~$\delta$, and let~$\hat\pi = \sem{normalize}(\pi)$.
  If $\dec{Check} \notin \hat\pi.\field{dec}$, then the evaluation of the body produces a view exactly the same as the normalized form.
  That is, $\Step*<\phase{Succ}>[\pi\{\fv{dec}{\pi.\field{dec} \setminus \{\dec{Check}\}},\fv{effq}{[]}\}, \sigma[x\mapsto v]]{e}{\_,\hat\pi, \omega}$ where~$\pi.\field{spec} = \<C,v\>$ and~$\delta[C] = \<\lfun x e, \sigma\>$.
\end{lem}

\begin{proof}
  The normalized view~$\hat\pi$ is the same as~$\pi$ except that the decision and state update queues are empty---$\dec{Check} \notin \hat\pi.\field{dec}$ implies also $\dec{Effect} \notin \hat\pi.\field{dec}$ from \cref{def:normalization}.
  Thus~$\hat\pi$ is stable as~$\pi$ is semi-stable.
  By \cref{def:stability}, evaluation with~$\hat\pi$ produces~$\hat\pi$ again, which also holds for retrying evaluation.
  Since~$\pi \approx \hat\pi$, by \cref{lem:simevalbody}, $\pi' = \hat\pi$.
\end{proof}

\begin{lem}[Stability of Retrying Evaluation]\label{lem:stab}
  Let~$\pi.\field{spec} = \<C,v\>$ and~$\delta[C] = \<\lambda x.e,\sigma\>$.
  If~$\pi$ is semi-stable and $\Step*[\pi,\sigma[x\mapsto v]]{e}{v',\pi', \omega}$,
  then~$\pi'$ is stable.
\end{lem}

\begin{proof}
  It is sufficient to show that the last evaluation of the component body produces a stable view.
  Since~$\dec{Check} \notin \pi'.\field{dec}$, \Rule{AppSetComp} is not included in the derivation and all state updates in~$\pi$ are pure.
  This means~$\sem{normalize}(\pi)$ has empty queues.
  Since~$\pi \approx \sem{normalize}(\pi)$, by the \cref{lem:stabnorm}, $\pi' = \sem{normalize}(\pi)$,
  which is stable.
\end{proof}

\begin{lem}[Preservation of Semi-Stability]\label{lem:semistab}
  If $\<e,\delta\> \smallstep^* \<t,m,\omega,\delta,\mu\>$, then~$m$ is $t$-semi-stable.
\end{lem}

\begin{proof}
  It is sufficient to show that the initial step generates a semi-stable tree memory, and derivation of each step preserves the semi-stability.
  \begin{description}
    \item[{\normalfont\Rule{StepInit}, \Rule{StepCheck}}]
      By \cref{lem:stab}, evaluation of each component body produces stable view.
      Updating views' children or adding decisions does not break the stability, so views modified by $\sem{init}$ and $\sem{visit}$ are always stable.
    \item[{\normalfont\Rule{StepEffect}, \Rule{StepEvent}}]
      The only rule modifying the view during normal evaluation is \Rule{AppSetNormal}.
      It only adds state updates or modify the decision of a view, which preserves semi-stability of the view. \qedhere
  \end{description}
\end{proof}

\begin{remark}
  While the semi-stability of memory is always preserved, the stability is often violated.
  When setter closures are applied during the execution of Effects or event handlers,
  the update functions are queued,
  which are flushed and applied during the next check.
  Note that the memory is indeed always stable after the check.
\end{remark}

\begin{restatable}[Preservation of Coherence]{lem}{coherence}\label{lem:coherence}
  If $\<e,\delta\> \smallstep^* \<t,m,\omega,\delta,\mu\>$, then~$m$ is $t$-coherent.
  That is, $\sem{init}$ and $\sem{visit}$ always produce~$t$-coherent tree memory.
\end{restatable}

\begin{proof}
  It is sufficient to show that the initial step generates a coherent tree memory, and the derivation of each step preserves the coherence.
  \begin{description}
    \item[{\normalfont\Rule{StepInit}}]
      A retrying evaluation of the component body produces a view whose decision does not include $\dec{Check}$ and state queues are empty, thus coherent.
      The result of $\sem{init}$ only contains new views.
      Therefore, the resulting memory $m$ is $t$-coherent.
    \item[{\normalfont\Rule{StepCheck}}]
      As stated above, the views whose bodies are evaluated are coherent.
      Views that are not evaluated are left as is in $m'$, thus preserving coherence.
      Therefore, each coherence of the views in the resulting memory $m'$ is preserved.
    \item[{\normalfont\Rule{StepEffect}, \Rule{StepEvent}}]
      The only rule modifying the view during normal evaluation is \Rule{AppSetNormal}. It adds a state update and adds $\dec{Check}$ decision to the view, which preserves the coherence of the view. \qedhere
  \end{description}
\end{proof}

\begin{lem}[Similar Reconciliations]\label{lem:reconsim}
  Reconciling with similar memories produces the equivalent results.
  That is, if $m \vdash \sem{reconcile}(t,s)=\<t',m', \omega\>$,
  then for $\hat m$ such that $m \approx_t \hat m$,
  $\hat m \vdash \sem{reconcile}(t,s) = \<t',\hat m', \omega\>$ and
  $m' \equiv_{t'} \hat m'$.
\end{lem}

\begin{proof}
  We proceed with induction on the derivation of $m \vdash \sem{reconcile}(t,s)=\<t',m', \omega\>$.
  \begin{description}
    \item[{\normalfont\Rule{ReconcileComNew}, \Rule{ReconcileOther}}]
      \sem{init}(s) does not read or modify existing views in $m$; it only adds new views to $m$.
      All new views are descendants of $t'$, and are the same in $m'$ and $\hat m'$.
      Since the evaluation contexts are identical, the output buffers are also the same.
    \item[{\normalfont\Rule{ReconcileList}}]
      Each derivation of $\sem{reconcile}(t_i,s_i)$ produces the same tree, the same output buffer, and $t'_i$-equivalent memories.
      Therefore, the final trees and the output buffers are the same and the final memories are $[\Overline{t'_i}]_{i=1}^n$-equivalent.
    \item[{\normalfont\Rule{ReconcileComEffect}}]
      The derivation of $\Bigstep*{\phi}{p}$ produces the same $v$, $\pi'$, and $\omega$ by \cref{lem:simevalbody} and $\sem{reconcile}(t,s)$ produces the same tree and output buffer with $t'$-equivalent memories by IH.
      Therefore, the final memories are also $t'$-equivalent. \qedhere
  \end{description}
\end{proof}

\begin{lem}[Similar Checks]\label{lem:simvisit}
  Checking with similar memories produces the equivalent results.
  That is, for similar memories $m$ and $\hat m$ that are $t$-coherent and $t$-semi-stable and have no view with \dec{Effect} decision set,
  if $m, \delta \vdash \sem{check}(t)=\<\mu,m', \omega\>$,
  then $\hat m, \delta \vdash \sem{check}(t) = \<\mu,\hat m', \omega'\>$ and
  $m' \equiv_{t'} \hat m'$.
  We also have $\omega = \omega'$ if component bodies do not print.
\end{lem}

\begin{proof}
  We proceed with induction on the derivation of $m, \delta \vdash \sem{check}(t)=\<\mu,m', \omega\>$.
  \begin{description}
    \item[{\normalfont\Rule{CheckConst}, \Rule{CheckClos}}] These cases are trivial.
    \item[{\normalfont\Rule{CheckList}}] ($t = [\Overline{t_i}]_{i=1}^n$).
      Each derivation of $\sem{check}(\hat t_i)$ produces memory, $t'_k$-equivalent to~$m_i$ for~$1 \le k \le i$ and $t'_k$-similar to $m_i$ for $i < k \le n$ by the IHs.
      As a result, $m_n$ and $\hat m_n$ are $t'_k$-equivalent for~$1 \le k \le n$, therefore $[\Overline{t'_i}]_{i=1}^n$-equivalent.
    \item[{\normalfont\Rule{CheckIdle}}] ($t = p$, $m[p] = \pi$, and $\hat m[p] = \hat \pi$).
      Note that $\pi$ is coherent and $\dec{Check} \notin \pi.\field{dec}$, implying that the state update queues are empty.
      Hence, $\sem{normalize}(\pi) = \pi$.
      \begin{itemize}
        \item $\dec{Check} \notin \hat \pi.\field{dec}$:
          By the IH, $\sem{check}(t)$ produces $t$-equivalent results.
        \item $\dec{Check} \in \hat \pi.\field{dec}$:
          By \cref{lem:stabnorm}, $\hat \pi' = \sem{normalize}(\hat \pi)$.
          Since $\pi \approx \hat \pi$, $\sem{normalize}(\hat \pi) = \sem{normalize}(\pi)$, we have~$\hat\pi' = \pi$.
          Since no view---including~$\pi$---has \dec{Effect} in~$m$, $\dec{Effect} \notin \hat \pi'.\field{dec}$ as well.
          Therefore, $\hat m, \delta \vdash \sem{check}(p) = \_$ is derived from \Rule{CheckNoEffect},
          and thus $\hat m, \delta \vdash \sem{check}(t) = \<b, \hat m', \omega\>$.
          By the IH, $m' \equiv_t \hat m'$.
          As a result, the results are $p$-equivalent.
          Note that we have $\omega = \omega'$ if the component prints nothing.
      \end{itemize}
    \item[{\normalfont\Rule{CheckNoEffect}}] ($t = p$, $m[p] = \pi$, and $\hat m[p] = \hat \pi$).
      \begin{itemize}
        \item $\dec{Check} \notin \hat \pi.\field{dec}$:
          The proof is similar to the case of \Rule{CheckIdle} where $\dec{Check} \in \hat m[p].\field{dec}$, except that $m$ and $\hat m$ are swapped.
        \item $\dec{Check} \in \hat \pi.\field{dec}$:
          By \cref{lem:simevalbody}, $\hat \pi' = \pi'$ (and $\omega = \omega'$ if componets print nothing),
          and by the IH, $m' \equiv_p \hat m'$.
          Therefore, the results are $p$-equivalent.
      \end{itemize}
    \item[{\normalfont\Rule{CheckEffect}}] ($t = p$).
      Since $\sem{normalize}(\hat \pi) = \sem{normalize}(\pi)$ and $\dec{Check} \in \sem{normalize}(\pi).\field{dec}$, $\dec{Check} \in \hat \pi.\field{dec}$.
      By \cref{lem:simevalbody}, $\pi' = \hat \pi'$,
      and by \cref{lem:reconsim}, $m' \equiv_p \hat m'$.
      Therefore, the results are $p$-equivalent. \qedhere
  \end{description}
\end{proof}

\simtransthm*

\begin{proof}
  Since the views not reachable from~$t$ are the same in~$m$ and~$\hat m$ before the transition,
  and these views remain unchanged, it suffices to show that the memories~$m'$ and~$\hat m'$ resulting from the transition from~$m$ and~$\hat m$, respectively, are $t$-equivalent.
  The transition from check mode $\checkmode$ is derived by the \sem{check} rules,
  so we can apply \cref{lem:simvisit}
  if we can show that the views in $\checkmode$ have no \dec{Effect} decision.
  Note that~$m$'s $t$-semi-stability and $t$-coherence have already been shown in \cref{lem:semistab,lem:coherence}.
  
  It also follows from \cref{lem:simvisit} that we have $\omega' = \omega''$ if component bodies do not print.

  We show that only views reachable from the root can contain an \dec{Effect} decision, and such views never contain an \dec{Effect} in check mode $\checkmode$ or event loop mode $\eloopmode$.
  \begin{description}
    \item[{\normalfont\Rule{StepInit}}] Every view added by \sem{init} is reachable from the root.
    \item[{\normalfont\Rule{StepEffect}}] The \sem{commitEffs} rules remove all \dec{Effect} decisions from reachable views.
    \item[{\normalfont\Rule{StepCheck}}] Only paths that are passed as arguments to \sem{check} or returned by \sem{reconcile} may have their views contain an \dec{Effect} decision,
    and those views are all reachable after \sem{check} or \sem{reconcile}.
    If any view receives an \dec{Effect} decision, the \sem{check} rule produces the rendered mode $\rendermode$, to which the state transitions.
    \item[{\normalfont\Rule{StepEvent}}] The \field{children} and \field{decision} fields are not modified during the \phase{Normal} phase. \qedhere
  \end{description}
\end{proof}

\clearpage
\section{Comparison of Reactive UI Frameworks}\label{app:sec:frameworks}
The categorization of reactive GUI web frameworks based on three properties---(a)~whether they re-read component specifications for re-rendering, (b)~how state updates are processed (queued or immediate), and (c)~which reactivity primitives they employ---are given in \cref{tab:frameworks}.

\begin{table}[ht]
  \centering
  \caption{Comparison of reactive UI frameworks.}\label{tab:frameworks}
  \smaller[2]
  \begin{tabular}{llll}
    \toprule
    Framework & Re-read Spec. & State Updates & Reactivity Primitives \\
    \midrule
    React & Yes & Queued & Hooks \\
    Preact & Yes & Queued & Hooks \\
    Dioxus & Yes & Immediate & Signals \\
    Solid & No & Immediate & Signals \\
    Leptos & No & Immediate & Signals \\
    Angular & No & Immediate & Signals \\
    Vue & No & Immediate & Signals \\
    Svelte w/ runes & No & Immediate & Signals \\
    Svelte w/o runes & No & Immediate & Compiler-assisted \\
    SwiftUI & Yes & Immediate & Compiler-assisted \\
    \bottomrule
  \end{tabular}
\end{table}

\DefineShortVerb{\+}
We compare various reactive UI frameworks to
\begin{enumerate}
  \item check if state updates are queued or immediate, and
  \item check if component logic gets re-evaluated every render.
\end{enumerate}
\subsection{React}
+counter+ gets printed every render.
+count+ updates are queued.
\begin{reactcode}
import { useState, useEffect } from 'react';
function Counter() {
  console.log('counter');
  const [count, setCount] = useState(0);
  useEffect(() => {
    if (count >= 3) {
      setCount(0);
    }
  }, [count]);
  function h() {
    console.log(count);
    setCount(count + 1);
    console.log(count);
  }
  return (
    <button onClick={h}>
      {count}
    </button>
  );
}
\end{reactcode}

\subsection{Preact}
+counter+ gets printed every render.
+count+ updates are queued.
\begin{reactcode}
import { useState, useEffect } from 'preact/hooks';
function Counter() {
  console.log('counter');
  const [count, setCount] = useState(0);
  useEffect(() => {
    if (count >= 3) {
      setCount(0);
    }
  }, [count]);
  function h() {
    console.log(count);
    setCount(count + 1);
    console.log(count);
  }
  return (
    <button onClick={h}>
      {count}
    </button>
  );
}
\end{reactcode}

\subsection{Dioxus}
+counter+ gets printed every render.
+count+ updates are immediate.
\begin{reactcode}
use dioxus::prelude::*;
fn log(msg: &str) { ... }
#[component]
fn Counter() -> Element {
  log("Counter");
  let mut count = use_signal(|| 0);
  use_effect(move || {
    if count() >= 3 {
      count.set(0);
    }
  });
  rsx! {
    button {
      onclick: move |_| {
        log(count.to_string().as_str());
        count += 1;
        log(count.to_string().as_str());
      },
      "{count}"
    }
  }
}
\end{reactcode}

\subsection{Solid}
+counter+ gets printed only once.
+count+ updates are immediate.
\begin{reactcode}
import { createSignal, createEffect } from 'solid-js';
function Counter() {
  console.log('counter');
  const [count, setCount] = createSignal(0);
  createEffect(() => {
    if (count() >= 3) {
      setCount(0);
    }
  });
  function h() {
    console.log(count());
    setCount(count() + 1);
    console.log(count());
  }
  return (
    <button onClick={h}>
      {count()}
    </button>
  );
}
\end{reactcode}

\subsection{Leptos}
+counter+ gets printed only once.
+count+ updates are immediate.
\begin{reactcode}
use leptos::leptos_dom::logging::*;
use leptos::prelude::*;
#[component]
pub fn Counter() -> impl IntoView {
  console_log("Counter");
  let (value, set_value) = signal(0);
  Effect::new(move |_| {
    if value.get() >= 3 {
      set_value.set(0);
    }
  });
  view! {
    <button on:click=move |_| {
      console_log(&value.get().to_string());
      set_value.update(|value| *value += 1);
      console_log(&value.get().to_string());
    }>{value}</button>
  }
}
\end{reactcode}

\subsection{Angular}
+counter+ gets printed only once.
+count+ updates are immediate.
\begin{reactcode}
import {Component, signal, effect} from '@angular/core';
@Component({
  selector: 'app-root',
  template: `
    <button (click)="h()">
      {{ count() }}
    </button>
  `,
})
export class CounterComponent {
  count = signal(0);
  constructor() {
    console.log('counter');
    effect(() => {
      if (this.count() >= 3) {
        this.count.set(0);
      }
    });
  }
  h() {
    console.log(this.count());
    this.count.update(val => val + 1);
    console.log(this.count());
  }
}
\end{reactcode}

\subsection{Vue}
+counter+ gets printed only once.
+count+ updates are immediate.
\begin{reactcode}
<script setup>
import { ref, watch } from 'vue';
console.log('counter');
const count = ref(0);
watch(count, (newValue) => {
  if (newValue >= 3) {
    count.value = 0;
  }
});
function h() {
  console.log(count.value);
  count.value += 1;
  console.log(count.value);
}
</script>
<template>
  <button @click="h">
    {{ count }}
  </button>
</template>
\end{reactcode}

\subsection{Svelte}
\subsubsection{Svelte with Runes}
+counter+ gets printed only once.
+count+ updates are immediate.
\begin{reactcode}
<script>
  console.log('counter');
  let count = $state(0);
  $effect(() => {
    if (count >= 3) {
      count = 0;
    }
  });
  function h() {
    console.log(count);
    count += 1;
    console.log(count);
  }
</script>
<button onclick={h}>
  {count}
</button>
\end{reactcode}

\subsubsection{Svelte without Runes}
+counter+ gets printed only once.
+count+ updates are immediate.
\begin{reactcode}
<svelte:options runes={false} />
<script>
  console.log('counter');
  let count = 0;
  $: if (count >= 3) {
    count = 0;
  }
  function h() {
    console.log(count);
    count += 1;
    console.log(count);
  }
</script>
<button on:click={h}>
  {count}
</button>
\end{reactcode}

\subsection{SwiftUI}
+counter+ gets printed every render.
+count+ updates are immediate.
\begin{reactcode}
struct Counter: View {
  @State private var count = 0
  var body: some View {
    print("counter")
    return Button("\(count)") {
      print(count)
      count += 1
      print(count)
    }.onChange(of: count) { oldValue, newValue in
      if newValue >= 3 {
        count = 0
      }
    }
  }
}
\end{reactcode}

\UndefineShortVerb{\+}

\end{document}